\pdfoutput=1
\documentclass{article}
\usepackage[numbers]{natbib}

\usepackage[preprint]{neurips_2023}

\usepackage[utf8]{inputenc} 
\usepackage[T1]{fontenc}    
\usepackage{url}            
\usepackage{booktabs}       
\usepackage{amsfonts}       
\usepackage{nicefrac}       
\usepackage{microtype}      
\usepackage[usenames,dvipsnames]{xcolor}

\usepackage[ruled,linesnumbered]{algorithm2e}
\usepackage{amsmath}
\usepackage{amsthm}
\usepackage{graphicx}
\usepackage{cleveref}
\usepackage{enumitem}
\usepackage{framed}
\usepackage{float,wrapfig}
\usepackage{subfig}
\usepackage{tcs-2}

\newcommand{\CH}{\texttt{CH}}
\newcommand{\CrudeNN}{\texttt{CrudeNN}}
\newcommand{\ChamferEstimate}{\texttt{Chamfer-Estimate}}

\newcommand{\boldeta}{\boldsymbol{\eta}}
\newcommand{\Prx}{\mathop{\Pr}}
\newcommand{\Ex}{\mathop{\E}}
\newcommand{\Varx}{\mathop{\Var}}

\title{A Near-Linear Time Algorithm for the Chamfer Distance}

\author{%
Ainesh Bakshi \\
MIT \\
\texttt{ainesh@mit.edu}
\And
Piotr Indyk \\
MIT \\
\texttt{indyk@mit.edu}
\And
Rajesh Jayaram \\
Google Research \\
\texttt{rkjayaram@google.com}
\And
Sandeep Silwal \\
MIT \\
\texttt{silwal@mit.edu}
\And
Erik Waingarten \\
University of Pennsylvania \\
\texttt{ewaingar@seas.upenn.edu}
}

\begin{document}

\maketitle

\begin{abstract}
For any two point sets $A,B \subset \R^d$ of size up to $n$, the Chamfer distance from $A$ to $B$ is defined as $\CH(A,B)=\sum_{a \in A} \min_{b \in B} d_X(a,b)$, where $d_X$ is the underlying distance measure (e.g., the Euclidean or Manhattan distance). The Chamfer distance is a popular measure of dissimilarity between point clouds, used in many machine learning, computer vision, and graphics applications, and admits a straightforward $\bigO{d n^2}$-time brute force algorithm. Further, the Chamfer distance is often used as a  proxy for the more computationally demanding Earth-Mover (Optimal Transport) Distance. However, the \emph{quadratic} dependence on $n$ in the running time  makes the naive approach intractable for large datasets.

We overcome this bottleneck and present the first $(1+\epsilon)$-approximate algorithm for estimating the Chamfer distance with a near-linear running time. Specifically, our algorithm runs in time $\bigO{nd \log (n)/\epsilon^2}$ and is implementable. Our experiments demonstrate that it is both accurate and fast on large high-dimensional datasets. We believe that our algorithm will open new avenues for analyzing large high-dimensional point clouds. We also give evidence that if the goal is to {\em report} a $(1+\eps)$-approximate mapping from $A$ to $B$ (as opposed to just its value), then any sub-quadratic  time algorithm is unlikely to exist.



\end{abstract}






\section{Introduction}

For any two point sets $A,B \subset \R^d$ of sizes up to $n$, the Chamfer distance\footnote{This is the definition adopted, e.g., in~\cite{athitsos2003estimating}. Some other papers, e.g.,~\cite{fan2017point}, replace each distance term $d_X(a,b)$ with its square, e.g., instead of $\|a-b\|_2$ they use $\|a-b\|_2^2$.  In this paper we focus on the first definition, as it emphasizes the connection to Earth Mover Distance and its relaxed weighted version in ~\cite{kusner2015word,atasu19a}.}  from $A$ to $B$ is defined as 

\begin{equation*}
\CH(A,B) = \sum_{a \in A} \min_{b \in B} d_X(a,b)
\end{equation*}

where $d_X$ is the underlying distance measure, such as the Euclidean or Manhattan distance. The Chamfer distance, and its weighted generalization called Relaxed Earth Mover Distance~\cite{kusner2015word,atasu19a}, are popular measures of dissimilarity between point clouds. They are widely used in 
machine learning (e.g.,~\cite{kusner2015word,wan2019transductive}), 
computer vision (e.g.,~\cite{athitsos2003estimating,sudderth2004visual,fan2017point,jiang2018gal})
and computer graphics~\cite{li2019lbs}.
Subroutines for computing Chamfer distances are available in popular libraries, such as Tensorflow~\cite{tensorflow}, Pytorch~\cite{pytorch3d} and PDAL~\cite{pdal}.
In many of those applications (e.g., \cite{kusner2015word}) Chamfer distance is used as a faster proxy for the more computationally demanding Earth-Mover (Optimal Transport) Distance. 

Despite the popularity of Chamfer distance, the naïve algorithm for computing it has quadratic $\bigO{n^2}$ running time, which makes it difficult to use for large datasets. Faster approximate algorithms can be obtained by performing $n$ exact or approximate nearest neighbor queries, one for each point in $A$. By utilizing the state of the art approximate nearest neighbor algorithms, this leads to $(1+\epsilon)$-approximate estimators with running times of $\bigO{n (1/\epsilon)^{\mathcal{O}(d)} \log n}$ in low dimensions~\cite{arya1998optimal} or roughly $\bigO{dn^{1+\frac{1}{2 (1+\epsilon)^2 -1}}}$ in high dimensions~\cite{andoni2015optimal}. Alas, the first bound suffers from exponential dependence on the dimension, while the second bound is significantly subquadratic only for relatively large approximation factors.



\subsection{Our Results}

In this paper we overcome this bottleneck and present the first $(1+\epsilon)$-approximate algorithm for estimating Chamfer distance that has a  {\em near-linear} running time, both in theory and in practice. Concretely, our contributions are as follows:
\begin{itemize}[leftmargin=*]
    \item When the underlying metric $d_X$ is defined by the $\ell_1$ or $\ell_2$ norm, we give an algorithm that runs in time $\bigO{nd \log (n)/\epsilon^2}$ and estimates the Chamfer distance up to $1\pm \eps$ with $99 \%$ probability (see Theorem~\ref{thm:estimating-chamfer-nearly-linear}). In general, our algorithm works for any metric $d_X$ supported by Locality-Sensitive Hash functions (see Definition~\ref{def:lsh-all-scale}), with the algorithm running time depending on the parameters of those functions. Importantly, the algorithm is quite easy to implement, see Figures~\ref{fig:alg-main} and \ref{fig:crude-nn}.
    \item For the more general problem of {\em reporting} a mapping $g:A \to B$ whose cost $\sum_{a \in A} d_X(a,g(a))$  is within a factor of $1+\eps$ from $\CH(A,B)$, we show that, under a popular complexity-theoretic conjecture, an algorithm with a running time analogous to that of our {\em estimation} algorithm does not exist, even when $d_X(a,b)=\|a-b\|_1$. 
    Specifically, under a Hitting Set Conjecture~\cite{williams2018some}, any such algorithm must run in time $\Omega(n^{2-\delta})$ for any constant $\delta>0$, even when the dimension $d=\Theta(\log^2 n)$ and $\eps=\frac{\Theta(1)}{d}$. (In contrast, our estimation algorithm runs in near-linear time for such parameters). This demonstrates that, for the Chamfer distance, estimation is significantly easier than reporting.
    \item We experimentally evaluate our algorithm on real and synthetic data sets. Our experiments demonstrate the effectiveness of our algorithm for both low and high dimensional datasets and across different dataset scales. Overall, it is much faster (\textbf{>5x}) than brute force (even accelerated with KD-trees) and both faster and more sample efficient (\textbf{5-10x}) than simple uniform sampling. We demonstrate the scalability of our method by running it on \emph{billion-scale} Big-ANN-Benchmarks datasets \cite{simhadri2022results}, where it runs up to \textbf{50x} faster than optimized brute force. 
    In addition, our method is robust to different datasets: while uniform sampling performs reasonably well for some datasets in our experiments, it performs poorly on datasets where the distances from points in $A$ to their neighbors in $B$ vary significantly. In such cases, our algorithm is able to adapt its importance sampling probabilities appropriately and obtain significant improvements over uniform sampling.  
\end{itemize}

\section{Algorithm and Analysis}\label{sec:algorithm}


In this section, we establish our main result for estimating Chamfer distance:

\begin{theorem}[Estimating Chamfer Distance in Nearly Linear Time]
\label{thm:estimating-chamfer-nearly-linear}
Given as input two datasets $A, B \subset \mathbb{R}^d$ such that $|A|, |B| \le n$, and an accuracy parameter $0< \eps<1$, \textup{$\ChamferEstimate$} runs in time $\bigO{nd\log(n)/\eps^2}$ and outputs an estimator $\eta$ such that with probability at least $99/100$,
\begin{equation*}
    \Paren{1-\eps} \emph{\CH}(A,B) \leq \eta \leq \Paren{ 1+\eps} \emph{\CH}(A,B),
\end{equation*}
when the underlying metric is Euclidean $(\ell_2)$ or Manhattan $(\ell_1)$ distance.
\end{theorem}


\begin{figure}[htb!]
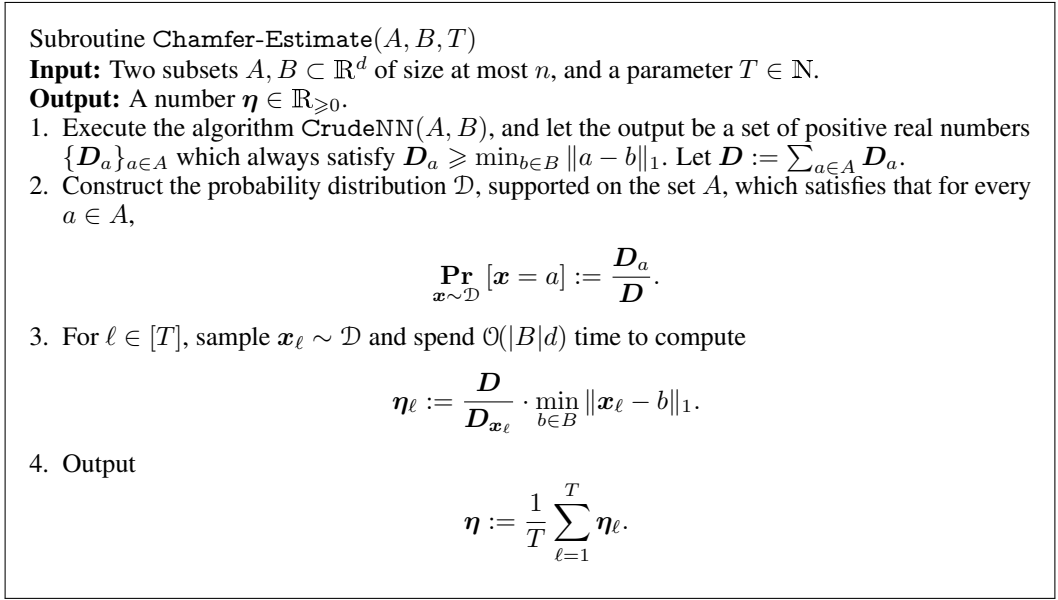

	\begin{framed}
		\noindent Subroutine $\ChamferEstimate(A, B, T)$
		
		\begin{flushleft}
			\noindent {\bf Input:} Two subsets $A, B \subset \R^d$ of size at most $n$, and a parameter $T \in \N$. 
			
			\noindent {\bf Output:} A number $\boldeta \in \R_{\geq 0}$. 
			
			\begin{enumerate}[leftmargin=*]
				\item\label{ln:first} Execute the algorithm $\CrudeNN(A,B)$, and let the output be a set of  positive real numbers $\{ \bD_a \}_{a \in A}$ which always satisfy $\bD_a \geq \min_{b \in B} \|a-b\|_1$. Let $\bD := \sum_{a \in A} \bD_a$. 
				\item\label{ln:two} Construct the probability distribution $\calD$, supported on the set $A$, which satisfies that for every $a \in A$,
				\begin{align*}
					\Prx_{\bx \sim \calD}\left[ \bx = a \right] := \frac{\bD_a}{\bD}. 
				\end{align*}
				\item\label{ln:three} For $\ell \in [T]$, sample $\bx_{\ell} \sim \calD$ and spend $\bigO{|B|d}$ time to compute
				\[ \boldeta_{\ell} := \frac{\bD}{\bD_{\bx_{\ell}}} \cdot \min_{b\in B} \| \bx_{\ell} - b\|_1. \]
				\item Output 
				\[ \boldeta :=  \frac{1 }{T} \sum_{\ell=1}^T \boldeta_{\ell}. \]
			\end{enumerate}
		\end{flushleft}
	\end{framed}
	\caption{The $\ChamferEstimate$ Algorithm.}\label{fig:alg-main}
\end{figure}

For ease of exposition, we make the simplifying assumption that the underlying metric is Manhattan distance, i.e. $d_X (a,b) =\| a-b\|_1$. Our algorithm still succeeds whenever the underlying metric admits a locality-sensitive hash function (see Definition~\ref{def:lsh-all-scale}). 

\paragraph{Uniform vs Importance Sampling.}
A natural algorithm for estimating $\CH(A,B)$ proceeds by \emph{uniform sampling}: sample an $a \in A$ uniformly at random and explicitly compute $\min_{b \in B} \|a-b\|_1$. In general, we can compute the estimator $\hat{z}$ for $\CH(A,B)$ by averaging over $s$ uniformly chosen samples, resulting in runtime $\bigO{nds}$. It is easy to see that the resulting estimator is un-biased, i.e. $\expecf{}{\hat{z}} = \CH(A,B)$. However, if a small constant fraction of elements in $A$ contribute significantly to $\CH(A,B)$, then $s = \Omega(n)$ samples could be necessary to obtain, say, a $1\%$ relative error estimate with constant probability. Since each sample requires a linear scan to find the nearest neighbor, this would result  in a quadratic runtime. 

While such an approach has good empirical performance for well-behaved datasets, it does not work for data sets where the distribution of the distances from points in $A$ to their nearest neighbors in $B$ is skewed.  Further, it is computationally prohibitive to verify the quality of the approximation given by uniform sampling. Towards proving Theorem~\ref{thm:estimating-chamfer-nearly-linear}, it is paramount to obtain an algorithm that works regardless of the structure of the input dataset.

A more nuanced approach is to perform \emph{importance sampling} where we sample $a \in A$ with probability proportional to its contribution to $\CH(A,B)$. In particular, if we had access to a distribution, $\bD_a$, over elements  $a\in A$ such that, $\min_{b\in B}\norm{a-b}_1 \leq \bD_a \leq\lambda \min_{b\in B}\norm{a-b}_1 $, for some parameter $\lambda>1$, then sampling $O\Paren{\lambda}$ samples results in an estimator $\hat{z}$ that is within $1\%$ relative error to the true answer with probability at least $99\%$. Formally, we consider the estimator defined in Algorithm~\ref{fig:alg-main}, where we assume access to $\CrudeNN(A, B)$, a sub-routine which receives as input $A$ and $B$ and outputs estimates $\bD_a \in \R_{\geq 0}$ for each $a \in A$ which is guaranteed to be an upper bound for $\min_{b \in B} \| a - b\|_1$. Based on the values $\{ \bD_a \}_{a \in A}$ we construct an importance sampling distribution $\calD$ supported on $A$. As a result, we obtain the following lemma:


\begin{lemma}[Variance Bounds for Chamfer Estimate]\label{lem:low-variance-estimate}
Let $n, d \in \N$ and suppose $A, B$ are two subsets of $\R^d$ of size at most $n$. For any $T\in \N$, the output $\boldeta$ of $\emph{\ChamferEstimate}(A, B, T)$ satisfies
\begin{align*}
\Ex\left[ \boldeta \right] &= \emph{\CH}(A, B), \\
\Varx\left[ \boldeta \right] &\leq \frac{1}{T} \cdot \emph{\CH}(A, B)^2 \left( \frac{\bD}{\emph{\CH}(A, B)} - 1 \right),
\end{align*}
for $\bD$ from Line~\ref{ln:first} in Figure~\ref{fig:alg-main}. The expectations and variance are over the randomness in the samples of Line~\ref{ln:three} of $\emph{\ChamferEstimate}(A, B, T)$. In particular, 
\begin{align*}
\Prx\Big[ \left| \boldeta - \emph{\CH}(A,B) \right| \geq \eps \cdot \emph{\CH}(A, B) \Big] \leq \frac{1}{\eps^2 \cdot T} \left(\frac{\bD}{\emph{\CH}(A, B)} - 1\right).
\end{align*}
\end{lemma}
The proof follows from a standard analysis of importance sampling and is deferred to Appendix~\ref{app:upper-bound}. Observe, if $\bD \leq \lambda \CH(A,B)$, it suffices to sample $T = O\Paren{\lambda/\eps^2}$ points in $A$, leading to a running time of $O\Paren{nd\lambda/\eps^2}$.

\begin{figure}[h!]
	\begin{framed}
		\noindent Subroutine $\CrudeNN(A,B)$
		
		\begin{flushleft}
			\noindent {\bf Input:} Two subsets $A,B$ of a metric space $(X, \|\cdot\|_1)$ of size at most $n$ such that all non-zero distances between any point in $A$ and any point in $B$ is between $1$ and $\poly(n/\eps)$. We assume access to a locality-sensitive hash family at every scale $\calH(r)$ for any $r \geq 0$ satisfying conditions of Definition~\ref{def:lsh-all-scale}. (We show in Appendix~\ref{app:upper-bound} that, for $\ell_1$ and $\ell_2$, the desired hash families exist, and that distances between $1$ and $\poly(n/\eps)$ is without loss of generality).
			
			\noindent {\bf Output:} A list of numbers $\{\bD_a\}_{a \in A}$ where $\bD_a \geq \min_{b\in B} \|a -b \|_1$. 
			
			\begin{enumerate}
				\item We instantiate $L = \bigO{\log (n/\eps)}$ and for $i \in \{ 0, \dots, L\}$, we let $r_i = 2^{i}$. 
				\item For each $i \in \{0, \dots, L\}$ sample a hash function $\bh_i \colon X \to U$ from $\bh_i \sim \calH(r_i)$.
				\item For each $a \in A$, find the smallest $i \in \{0,\dots, L\}$ for which there exists a point $b \in B$ with $\bh_i(a) = \bh_i(b)$, and set $\bD_a = \|a-b \|_1$.
				\begin{itemize}
					\item The above may be done by first hashing each point $b \in B$ and $i \in \{0,\dots, L\}$ according to $\bh_i(b)$. Then, for each $a \in A$, we iterate through $i \in \{0, \dots, L\}$ while hashing $a$ according to $\bh_i(a)$ until the first $b \in B$ with $\bh_i(a) = \bh_i(b)$ is found.
				\end{itemize}
			\end{enumerate}
		\end{flushleft}
	\end{framed}
	\caption{The $\CrudeNN$ Algorithm.}\label{fig:crude-nn}
\end{figure}

\paragraph{Obtaining importance sampling probabilities.}
It remains to show how to implement the $\CrudeNN(A, B)$ subroutine to obtain the distribution over elements in $A$ which is a reasonable over-estimator of the true probabilities. 
A natural first step is to consider performing an $\bigO{\log n}$-approximate nearest neighbor search (NNS): for every $a' \in A$, find $b' \in B$ satisfying $\|a' - b'\|_1/\min_{b \in B} \|a' - b\|_1 = \bigO{\log n}$. This leads to the  desired guarantees on $\{\bD_a\}_{a \in A}$. Unfortunately, the state of the art algorithms for $\bigO{\log n}$-approximate NNS, even under the $\ell_1$ norm, posses extraneous $\mbox{poly}(\log n)$ factors in the runtime, resulting in a significantly higher running time. These factors are even higher for the $\ell_2$ norm. 
Therefore, instead of performing a direct reduction to approximate NNS, we open up the approximate NNS black-box and give a simple algorithm which directly satisfies our desired guarantees on $\{\bD_a\}_{a \in A}$.

To begin with, we assume that the aspect ratio of all pair-wise distances is bounded by a fixed polynomial, $\poly(n/\eps)$ (we defer the reduction from an arbitrary input to one with polynomially bounded aspect ratio to Lemma~\ref{lem:polynomial-aspect-ratio-suffices}).
We proceed via computing  $\bigO{\log(n/\eps)}$ different (randomized) partitions of the dataset $A \cup B$. The $i$-th partition, for $1 \le i \le \bigO{\log (n/\eps)}$, can be written as $A \cup B = \cup_j \mathcal{P}^i_j$ and approximately satisfies the property that points in $A \cup B$ that are at distance at most $2^i$ will be in the same partition $\mathcal{P}^i_j$ with sufficiently large probability. To obtain these components, we use a family of \emph{locality-sensitive hash functions}, whose formal properties are given in Definition \ref{def:lsh-all-scale}. Intuitively, these hash functions guarantee that:
\begin{enumerate}
    \item For each $a' \in A$, its \emph{true} nearest neighbor $b' \in B$ falls into the \emph{same} component as $a'$ in the $i_0$-th partition, where $2^{i_0} = \Theta(\|a' - b'\|_1$) \footnote{Recall we assumed all distances are between $1$ and $\poly(n)$ resulting in only $\bigO{\log n}$ different partitions}, and
    \item Every other extraneous $b \ne b'$ is \emph{not} in the same component as $a'$ for each $i < i_0$.
\end{enumerate}

It is easy to check that any hash function that satisfies the aforementioned guarantees yields a valid set of distances $\{\bD_a\}_{a \in A}$ as follows: for every $a' \in A$, find the smallest $i_0$ for which there exists a $b' \in B$ in the same component as $a'$ in the $i_0$-th partition. Then set $\bD_{a'} = \|a'-b'\|_1$. Intuitively, the $b'$ we find for any fixed $a'$ in this procedure will have distance that is at least the closest neighbor in $B$ and with good probability, it won't be too much larger. A caveat here is that we cannot show the above guarantee holds for $2^{i_0} =  \Theta(\|a' - b'\|_1)$.  Instead, we obtain the slightly weaker guarantee that, {\em in the expectation}, the partition $b'$ lands in is a $\bigO{\log n}$-approximation to the minimum distance, i.e.   $2^{i_0} =  \Theta(\log n \cdot \|a' - b'\|_1)$. Therefore, after running $\CrudeNN(A, B)$, setting $\lambda=\log n$ suffices for our $\bigO{nd \log(n)/\eps^2}$ time algorithm. We formalize this argument in the following lemma:

\begin{lemma}[Oversampling with bounded Aspect Ratio]
\label{lem:intro-oversampling}
Let $(X, d_X)$ be a metric space with a locality-sensitive hash family at every scale (see Definition~\ref{def:lsh-all-scale}). Consider two subsets $A, B \subset X$ of size at most $n$ and any $\eps \in (0, 1)$ satisfying
\[ 1 \leq \min_{\substack{a \in A, b \in B \\ a \neq b}} d_X(a,b) \leq \max_{\substack{a \in A, b \in B}} d_X(a,b) \leq \poly(n/\eps). \]
Algorithm~\ref{fig:crude-nn}, $\emph{\CrudeNN}(A,B)$, outputs a list of (random) positive numbers $\{ \bD_a \}_{a \in A}$ which satisfy the following two guarantees:
\begin{itemize}
\item With probability $1$, every $a \in A$ satisfies $\bD_a \geq \min_{b \in B} d_X(a, b)$. 
\item For every $a \in A$,  $\Ex[\bD_a] \leq \bigO{\log n} \cdot \min_{b\in B}d_X(a,b)$.
\end{itemize}
Further, Algorithm~\ref{fig:crude-nn}, runs in time $\bigO{dn\log(n/\eps)}$ time, assuming that each function used in the algorithm can be evaluated in $\bigO{d}$ time.
\end{lemma}

\begin{proof}[Proof Sketch for Theorem~\ref{thm:estimating-chamfer-nearly-linear}]

Given the lemmas above, it is straight-forward to complete the proof of 
Theorem~\ref{thm:estimating-chamfer-nearly-linear}. 
First, we reduce to the setting where the aspect ratio is $\poly(n/\eps)$ (see Lemma~\ref{lem:polynomial-aspect-ratio-suffices} for a formal reduction).
We then invoke Lemma~\ref{lem:intro-oversampling} and apply Markov's inequality to obtain a set of distances $\bD_a$ such that with probability at least $99/100$, for each $a \in A$, $\min_{b\in B} \| a- b \|_1 \leq  \bD_a $ 
and $\sum_{a\in A} \bD_a \leq \bigO{\log(n)} \CH\Paren{A, B}$. 
We then invoke Lemma~\ref{lem:low-variance-estimate} and set the number of samples, $T=\bigO{\log(n)/\eps^2}$. 
The running time of our algorithm is then given by the time of $\CrudeNN(A, B)$, which is $O(nd \log(n/\eps))$, and the time needed to evaluate the estimator in Lemma~\ref{lem:low-variance-estimate}, requiring $\bigO{ nd\log(n)/\eps^2 }$ time. Refer to Section~\ref{app:upper-bound} for the full proof. 
\end{proof}

\section{Experiments}

We perform an empirical evaluation of our Chamfer distance estimation algorithm.

\paragraph{Summary of Results} 
Our experiments demonstrate the effectiveness of our algorithm for both low and high dimensional datasets and across different dataset sizes. Overall, it is much faster than brute force (even accelerated with KD-trees). Further, our algorithm is both faster and more sample-efficient than uniform sampling. It is also robust to different datasets: while uniform sampling performs well for most datasets in our experiments, it performs poorly on datasets where the distances from points in $A$ to their neighbors in $B$ vary significantly. In such cases, our algorithm is able to adapt its importance sampling probabilities appropriately and obtain significant improvements over uniform sampling.

\begin{table}[!ht]
\centering
{\renewcommand{\arraystretch}{1.3}
\begin{tabular}{c|c|c|c|c|c}
\textbf{Dataset} & \textbf{$|A|, |B|$}                             & \textbf{$d$} & \textbf{Experiment} & \textbf{Metric} & \textbf{Reference} \\ \hline
ShapeNet         & $\sim 8 \cdot 10^3, \sim 8 \cdot 10^3$    & $3$          & Small Scale         & $\ell_1$        &       \cite{chang2015shapenet}             \\
Text Embeddings  & $ 2.5\cdot 10^3, 1.8 \cdot 10^3$ & $300$        & Small Scale         & $\ell_1$        &     \cite{kusner2015word}          \\
Gaussian Points  & $5 \cdot 10^4, 5 \cdot 10^4$                    & $2$          & Outliers            & $\ell_1$        &     -               \\
DEEP1B           & $10^4, 10^9$                                    & $96$         & Large Scale         & $\ell_2$        &      \cite{babenko2016efficient}              \\
Microsoft-Turing & $10^5, 10^9$                                    & $100$        & Large Scale         & $\ell_2$        &     \cite{simhadri2022results}              
\end{tabular}
}
\caption{Summary of our datasets. For ShapeNet, the value of $|A|$ and $|B|$ is averaged across different point clouds in the dataset.\label{table:dataset}}
\end{table}

\subsection{Experimental Setup}

We use three different experimental setups, small scale, outlier, and large scale. They are designed to `stress test' our algorithm, and relevant baselines, under vastly different parameter regimes. The datasets we use are summarized in Table \ref{table:dataset}. For all experiments, we introduce uniform sampling as a competitive baseline for estimating the Chamfer distance, as well as (accelerated) brute force computation. All results are averaged across $20+$ trials and $1$ standard deviation error bars are shown when relevant.

\paragraph{Small Scale} 
These experiments are motivated from common use cases of Chamfer distance in the computer vision and NLP domains. In our small scale experiments, we use two different datasets: (a) the ShapeNet dataset, a collection of point clouds of objects in three dimensions \cite{chang2015shapenet}. ShapeNet is a common benchmark dataset frequently used in computer graphics, computer vision, robotics and Chamfer distance is a widely used measure of similarity between different ShapeNet point clouds \cite{chang2015shapenet}. 
(b) We create point clouds of words from text documents from \cite{kusner2015word}. Each point represents a word embedding obtained from the word-to-vec model of \cite{mikolov2013distributed} in $\R^{300}$ applied to the Federalist Papers corpus.  As mentioned earlier, a popular relaxation of the common Earth Mover Distance is exactly the (weighted) version of the Chamfer distance \cite{kusner2015word,atasu19a}.

Since ShapenNet is in three dimensions, we implement nearest neighbor queries using KD-trees to accelerate the brute force baseline  as KD-trees can perform exact nearest neighbor search quickly in small dimensions. However, they have runtime exponential in dimension meaning they cannot be used for the text embedding dataset, for which we use a standard naive brute force computation. For both these datasets, we implement our algorithms using Python 3.9.7 on an M1 MacbookPro with 32GB of RAM. We also use an efficient implementation of KD trees in Python and use Numpy and Numba whenever relevant. Since the point clouds in the dataset have approximately the same $n$ value, we compute the symmetric version $\CH(A,B) + \CH(B,A)$. For these experiments, we use the $\ell_1$ distance function.

\paragraph{Outliers} This experiment is meant to showcase the robustness of our algorithm. We consider two point clouds, $A$ and $B$, each sampled from Gaussian points in $\R^{100}$ with identity covariance. Furthermore, we add an "outlier" point to $A$ equal to $0.5n \cdot \textbf{1}$, where $\textbf{1}$ is the all ones vector.

This example models scenarios where the distances from points in $A$ to their nearest neighbors in $B$ vary significantly, and thus uniform sampling might not accurately account for all distances, missing a small fraction of large ones. 

\paragraph{Large Scale}
The purpose of these experiments is to demonstrate that our method scales to datasets with billions of points in hundreds of dimensions. We use two challenging approximate nearest neighbor search datasets: DEEP1B  \cite{babenko2016efficient} and Microsoft Turing-ANNS   \cite{simhadri2022results}. For these datasets, the set $A$ is the query data associated with the datasets. Due to the asymmetric sizes, we compute $\CH(A,B)$. These datasets are normalized to have unit norm and we consider the $\ell_2$ distance function.

These datasets are too large to handle using the prior configurations. Thus, we use a proprietary in-memory parallel implementation of the SimHash algorithm, which is an $\ell_2$ LSH family for normalized vectors according to Definition \ref{def:lsh-all-scale} \cite{charikar2002similarity}, on a 
shared virtual compute cluster with 2x64 core AMD Epyc 7763 CPUs (Zen3) with 2.45Ghz - 3.5GHz clock frequency, 2TB DDR4 RAM and 256 MB L3 cache. We also utilize parallization on the same compute cluster for naive brute force search.

\subsection{Results}

\paragraph{Small Scale}

\begin{figure}[!ht]
    \centering
    \subfloat[\centering ShapeNet]{{\includegraphics[width=4.65cm]{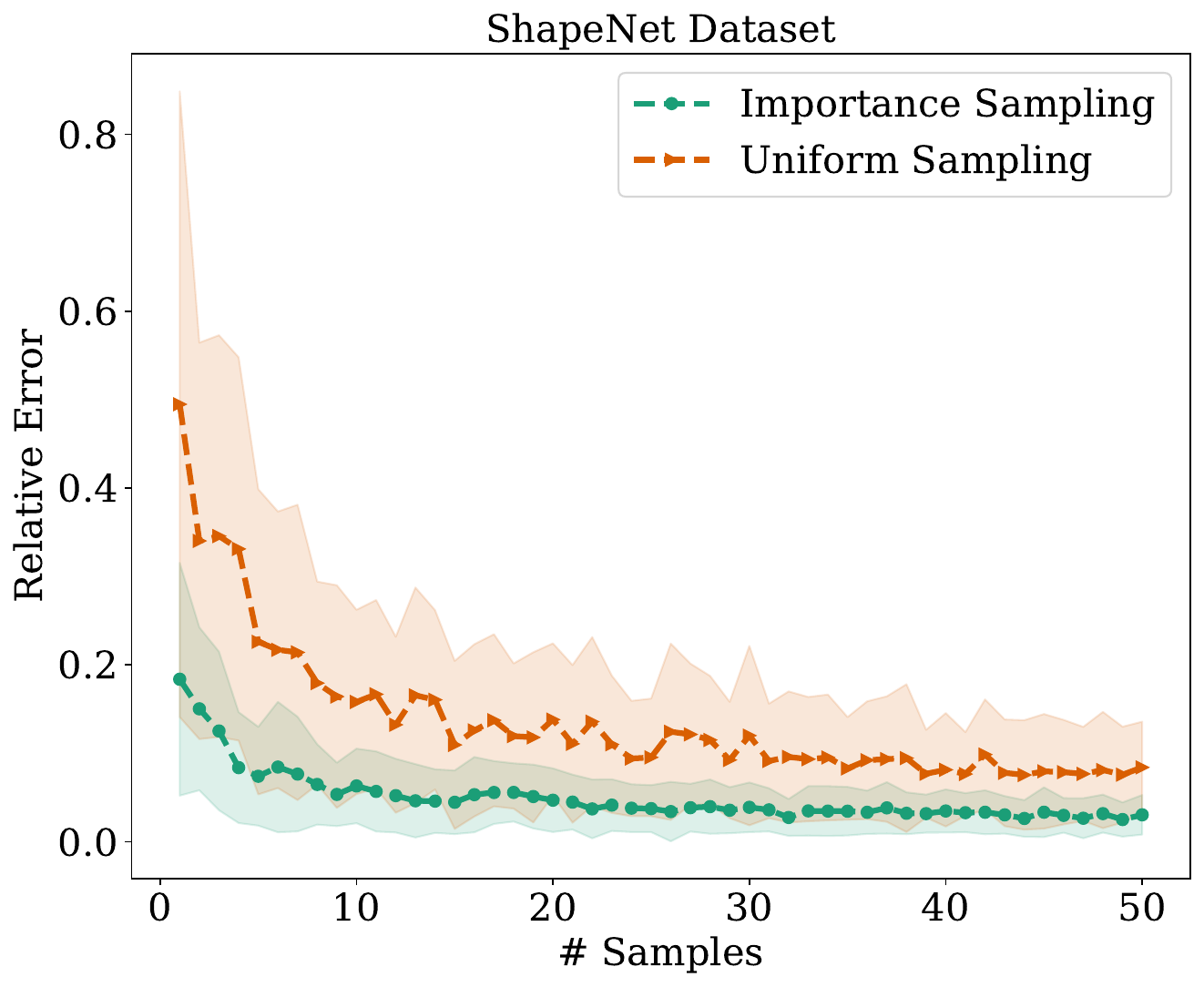}}}%
    \subfloat[\centering Federalist Papers]{{\includegraphics[width=4.65cm]{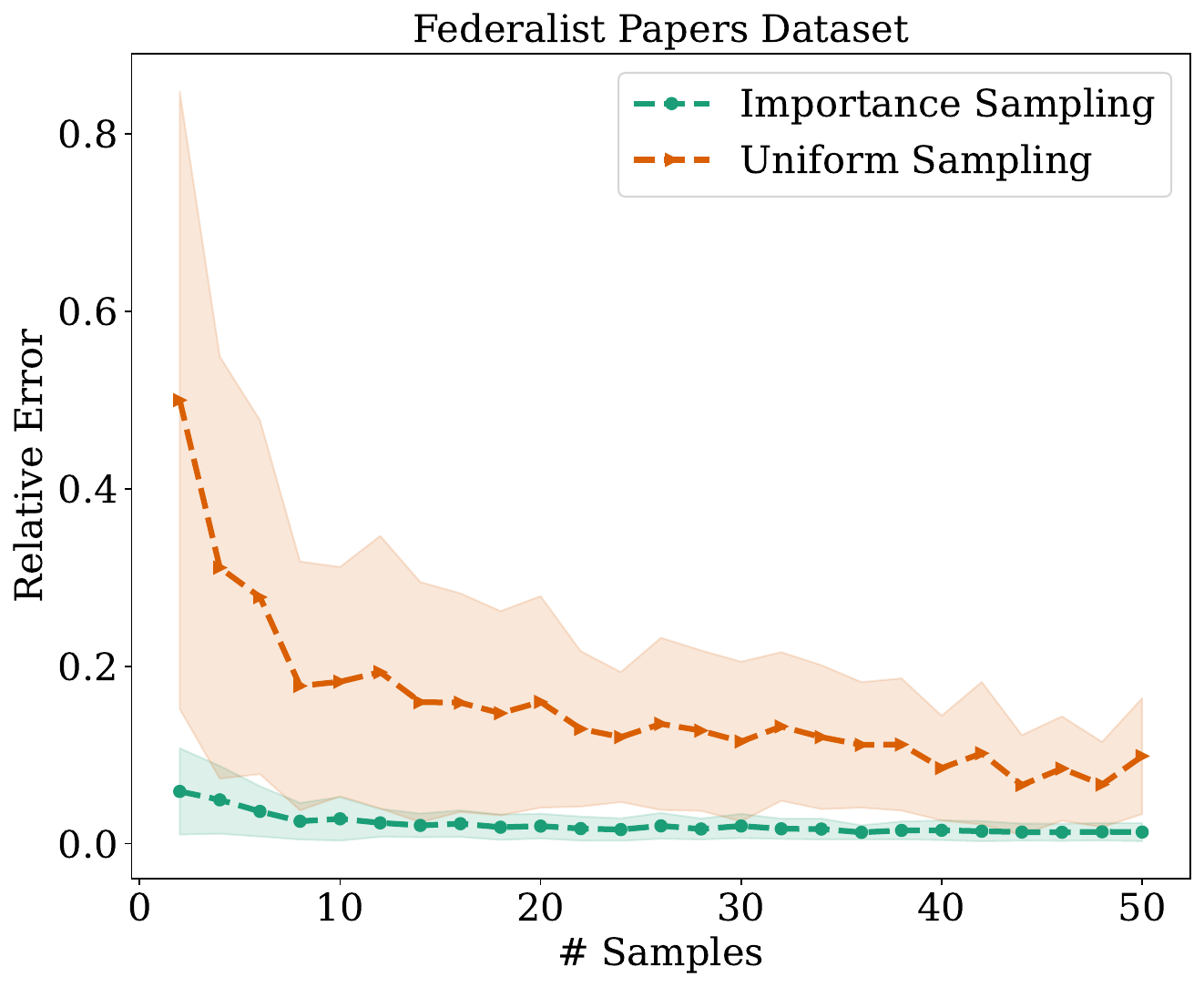} }}%
    \subfloat[\centering Gaussian Points]{{\includegraphics[width=4.65cm]{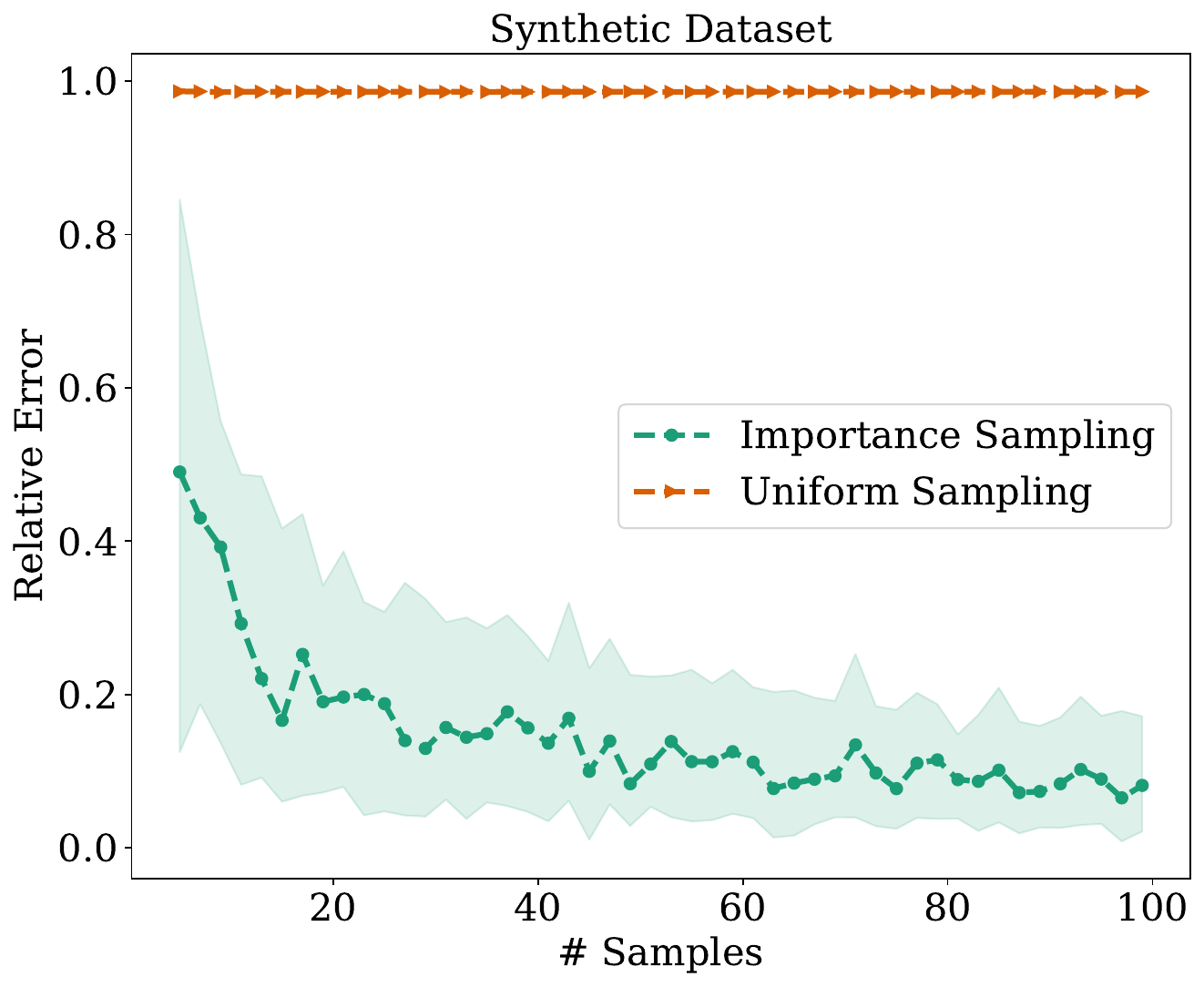} }}%
    \\
    \subfloat[\centering DEEP]{{\includegraphics[width=4.65cm]{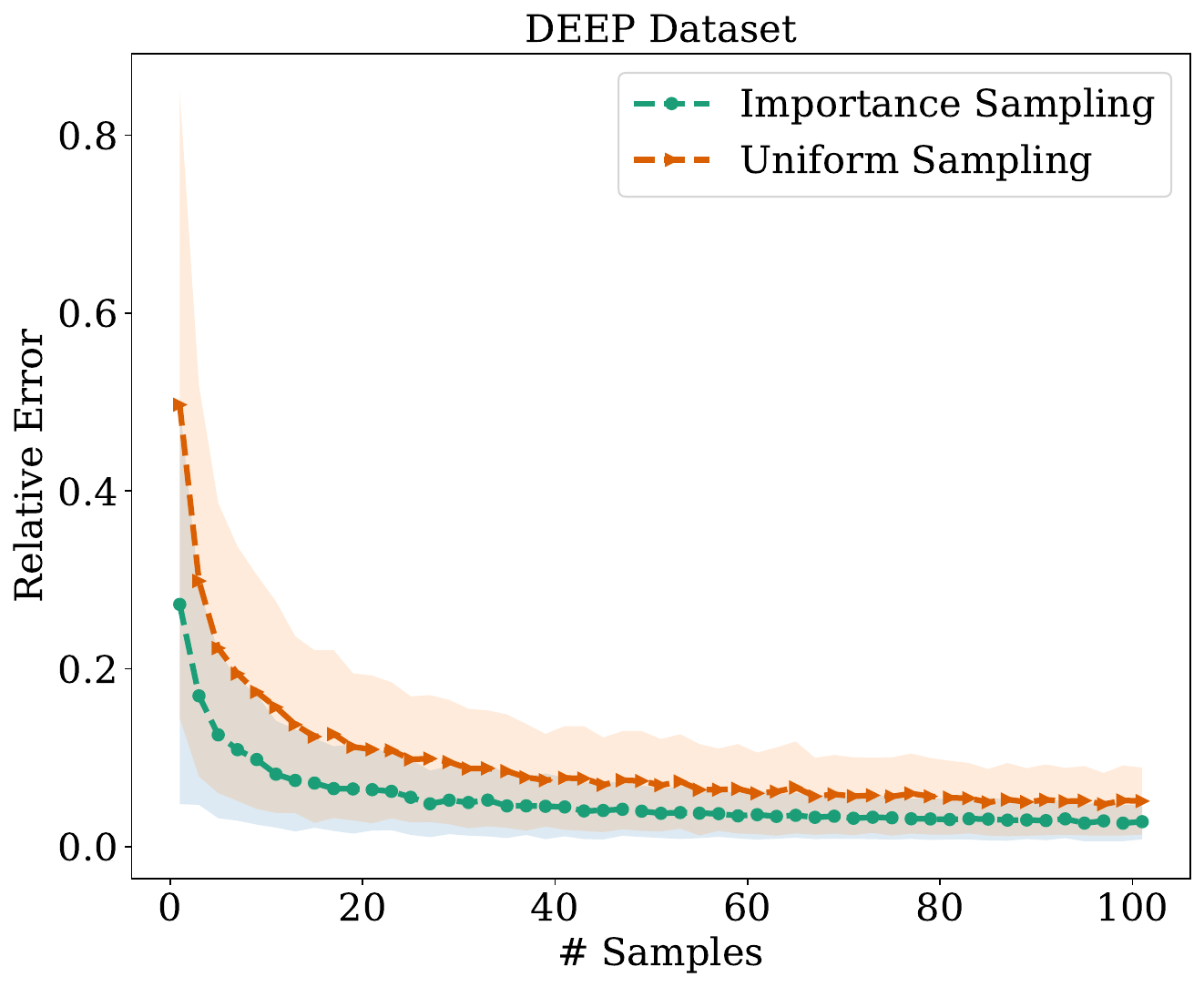} }}%
    \subfloat[\centering Turing]{{\includegraphics[width=4.65cm]{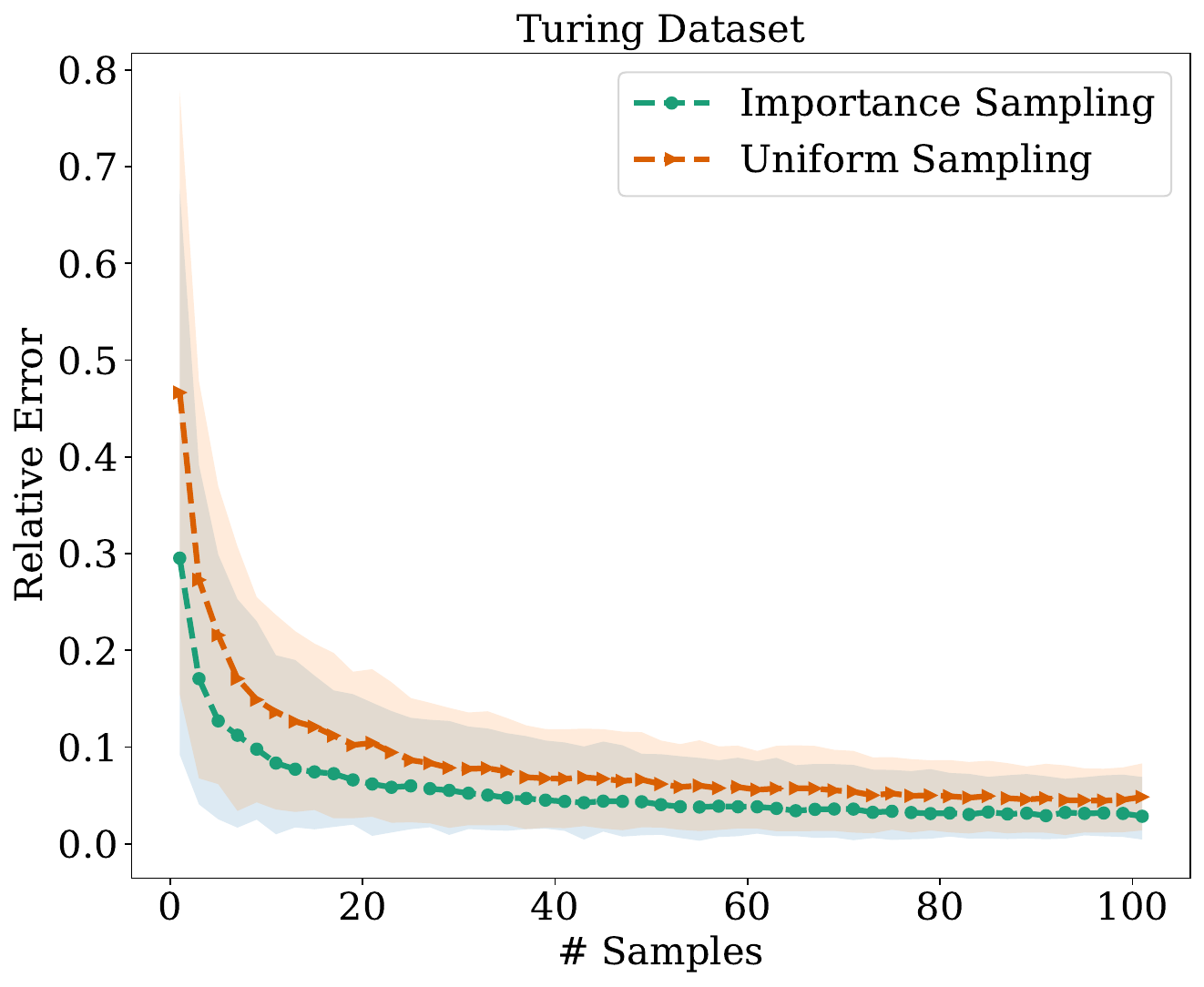} }}%
    \caption{Sample complexity vs relative error curves.}%
    \label{fig:sample_complexity1}%
\end{figure}

\begin{figure}[!ht]
    \centering
    \subfloat[\centering ShapeNet]{{\includegraphics[width=5cm]{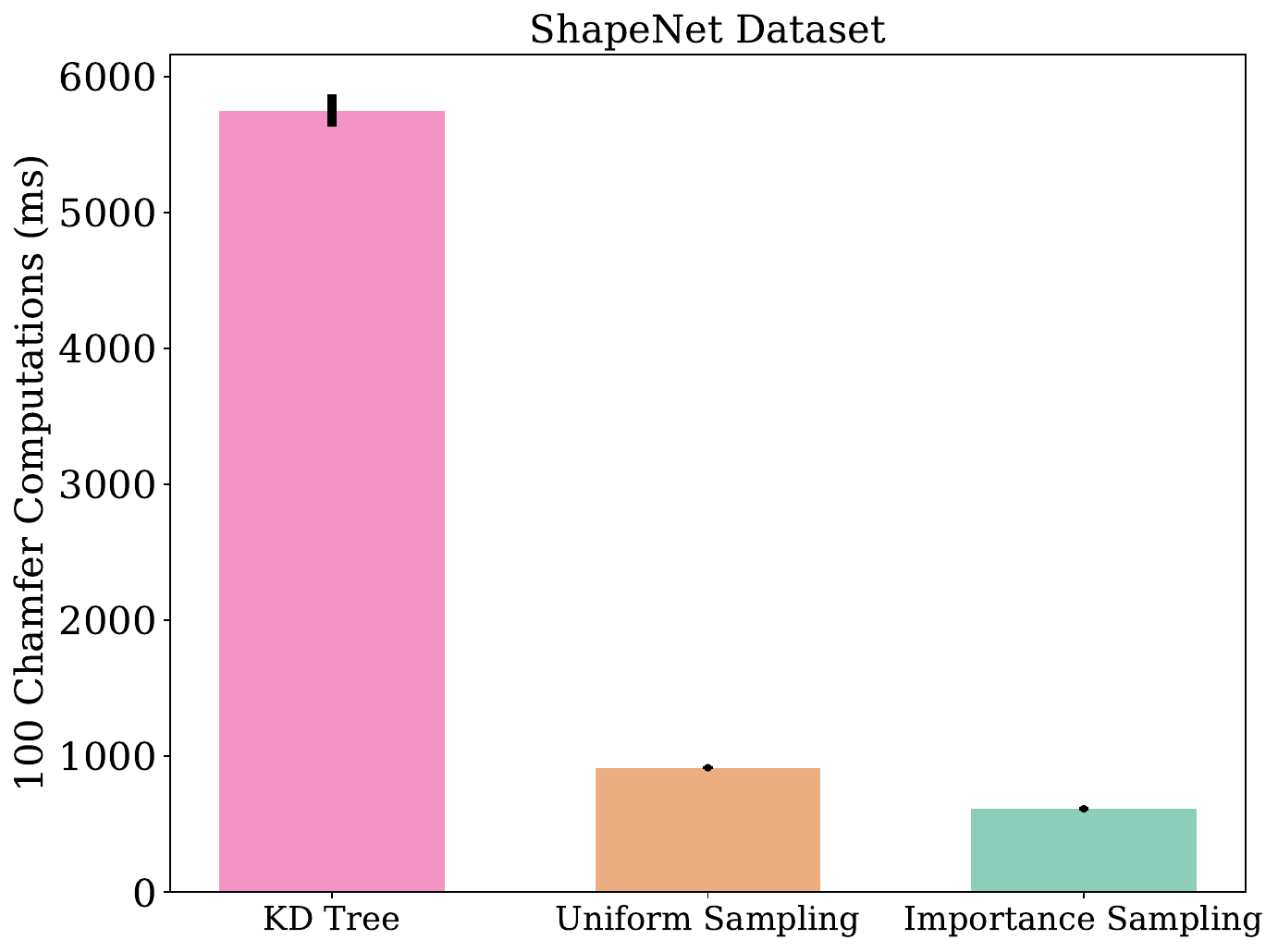} }}%
    \qquad
    \subfloat[\centering Federalist Papers]{{\includegraphics[width=5cm]{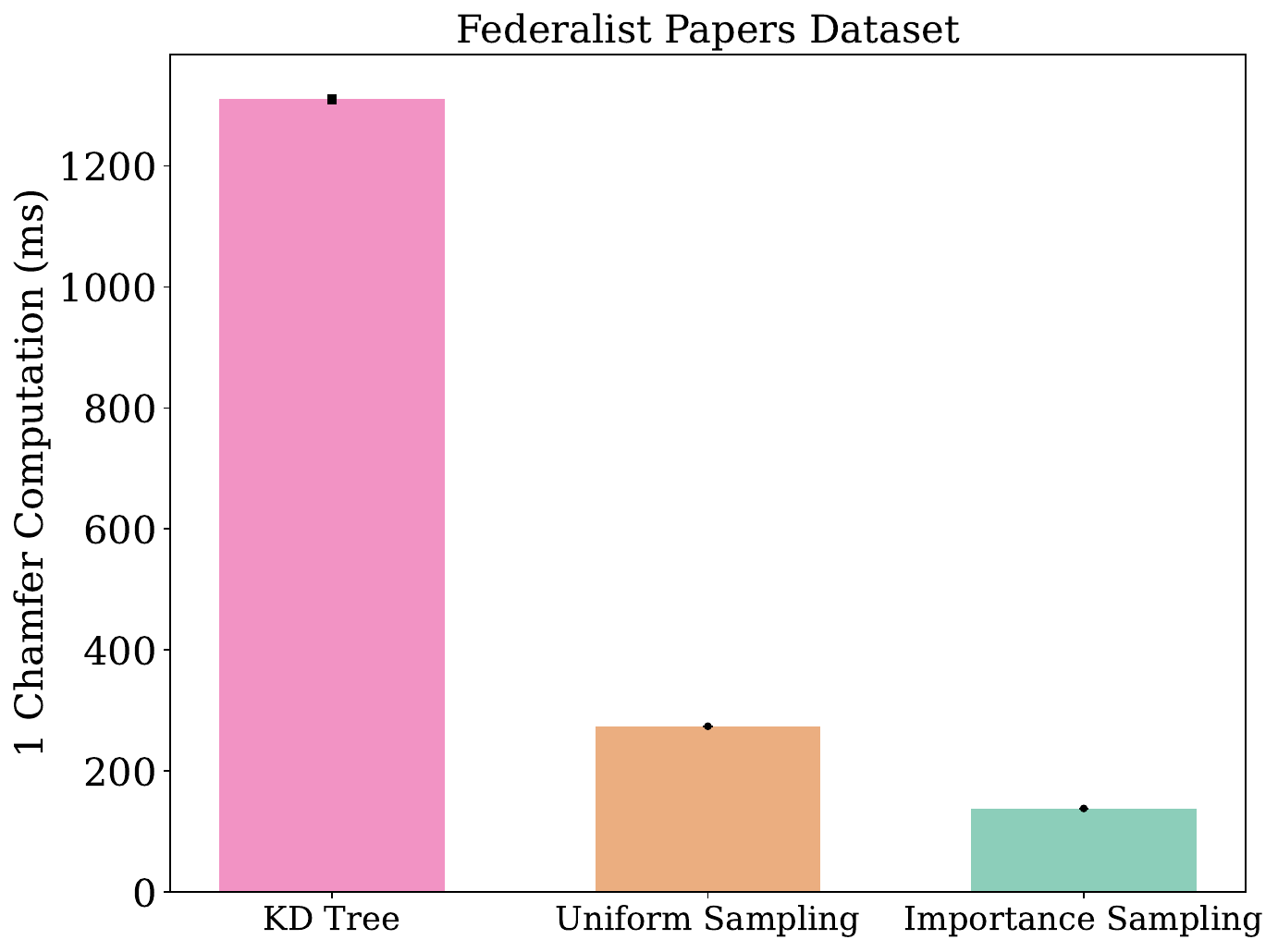} }}%
    \caption{Runtime experiments. We set the number of samples for uniform and importance sampling such that the relative errors of their respective approximations are similar.}%
    \label{fig:runtime_1}%
\end{figure}

\begin{figure}[!ht]
    \centering
    \subfloat[\centering DEEP]{{\includegraphics[width=5cm]{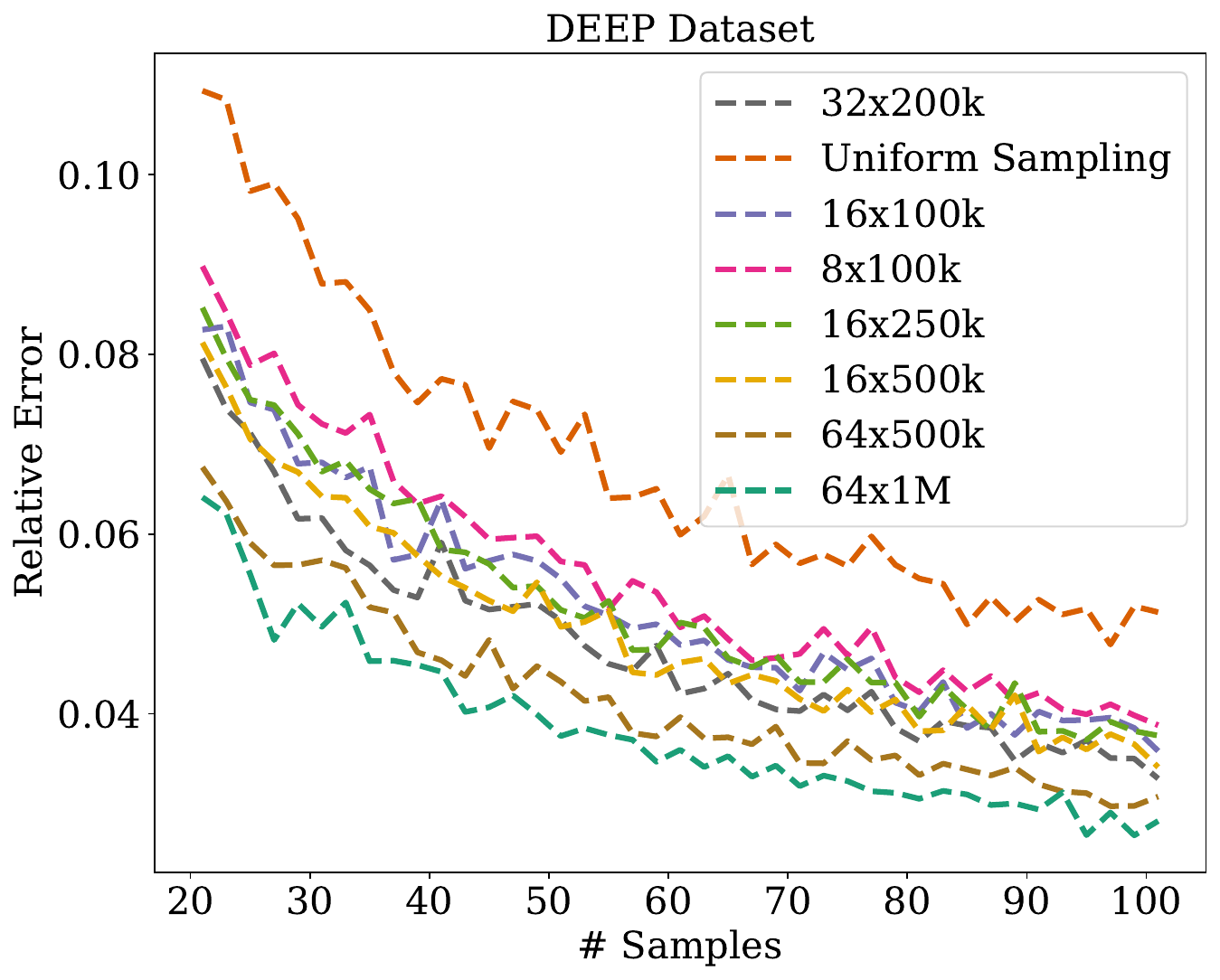} }}%
    \qquad
    \subfloat[\centering Turing]{{\includegraphics[width=5cm]{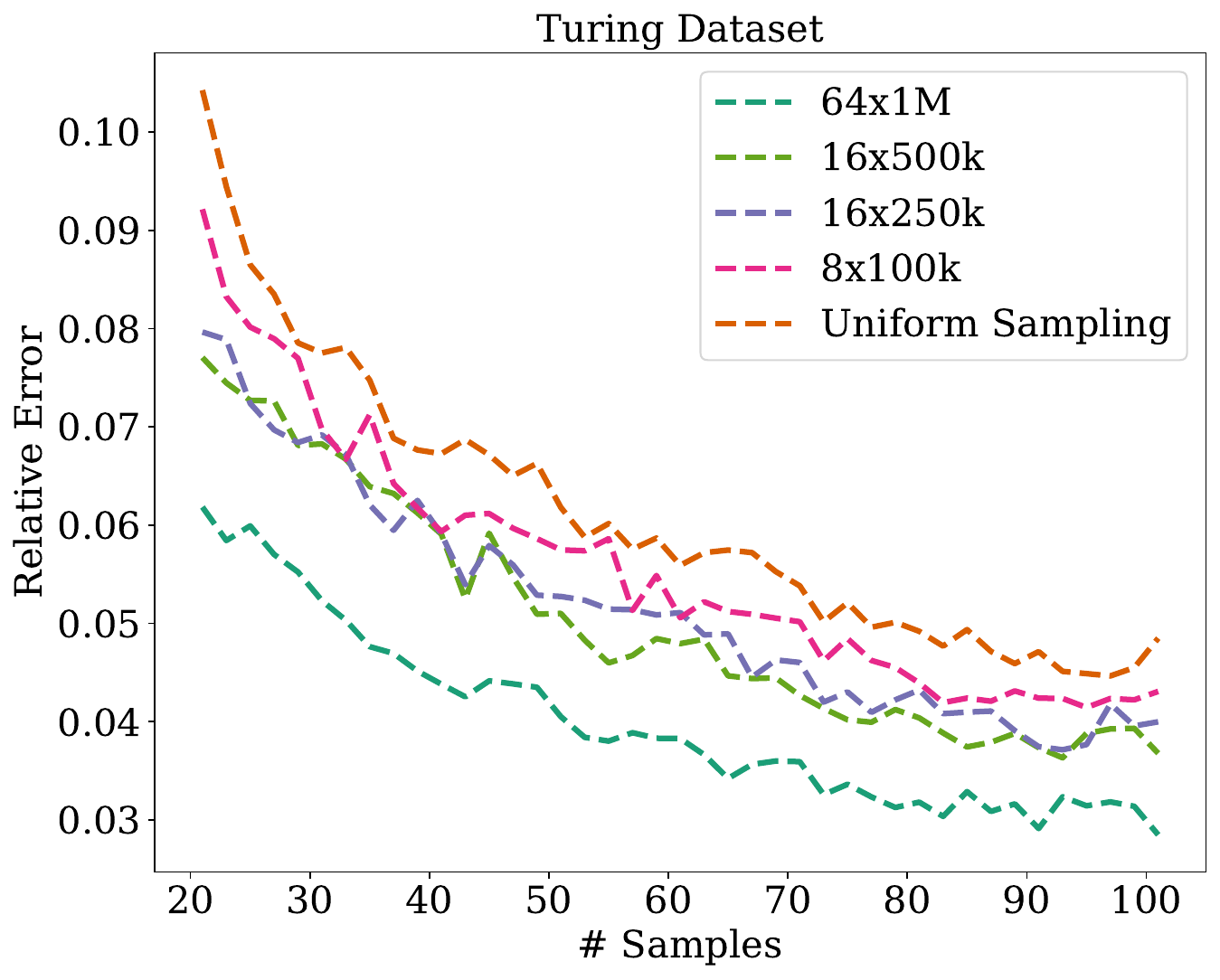} }}%
    \caption{The figures show sample complexity vs relative error curves as we vary the number of LSH data structures and window sizes. Each curve maps $k\times W$ where $k$ is the number of LSH data structures we use to repeatedly hash points in $B$ and $W$ is the window size, the number of points retrieved from $B$ that hash closest to any given $a$ at the smallest possible distance scales.\vspace{-5mm}}%
    \label{fig:more_largescale}%
\end{figure}

First we discuss configuring parameters. Recall that in our theoretical results, we use $\bigO{\log n}$ different scales of the LSH family in $\CrudeNN$. $\CrudeNN$ then computes (over) estimates of the nearest neighbor distance from points in $A$ to $B$ (in near linear time) which is then used for importance sampling by $\ChamferEstimate$. Concretely for the $\ell_1$ case, this the LSH family corresponds to imposing $\bigO{\log n}$ grids with progressively smaller side lengths. In our experiments, we treat the number of levels of grids to use as a tuneable parameter in our implementation and find that a very small number suffices for high quality results in the importance sampling phase.

Figure \ref{fig:more_shapenet} (b) shows that only using $3$ grid levels is sufficient for the crude estimates $\bD_a$ to be within a factor of $2$ away from the true nearest neighbor values for the ShapeNet dataset, averaged across different point clouds in the dataset. Thus for the rest of the Small Scale experiments, we fix the number of grid levels to be $3$.

Figure \ref{fig:sample_complexity1} (a) shows the sample complexity vs accuracy trade offs of our algorithm, which uses importance sampling, compared to uniform sampling. Accuracy is measured by the relative error to the true value. We see that our algorithm possesses a better trade off as we obtain the same relative error using only $10$ samples as uniform sampling does using $50+$ samples, resulting in at least a $\textbf{5x}$ improvement in sample complexity. For the text embedding  dataset, the performance gap between our importance sampling algorithm and uniform sampling grows even wider, as demonstrated by Figure \ref{fig:sample_complexity1} (b), leading to $\textbf{> 10x}$ improvement in sample complexity.

In terms of runtimes, we expect the brute force search to be much slower than either importance sampling and uniform sampling. Furthermore, our algorithm has the overhead of first estimating the values $\bD_a$ for $a \in A$ using an LSH family, which uniform sampling does not. However, this is compensated by the fact that our algorithm requires much fewer samples to get accurate estimates.

Indeed, Figure \ref{fig:runtime_1} (a) shows the average time of 100 Chamfer distance computations between randomly chosen pairs of point clouds in the ShapeNet dataset. We set the number of samples for uniform sampling and importance sampling (our algorithm) such that they both output estimates with (close to) $2\%$ relative error. Note that our runtime includes the time to build our LSH data structures. This means we used $100$ samples for importance sampling and $500$ for uniform.  The brute force KD Tree algorithm (which reports exact answers) is approximately 5x slower than our algorithm. At the same time, our algorithm is $50\%$ faster than uniform sampling. For the Federalist Papers dataset (Figure \ref{fig:runtime_1} (b)), our algorithm only required $20$ samples to get a $2\%$ relative error approximation, whereas uniform sampling required at least $450$ samples. As a result, our algorithm achieved \textbf{2x} speedup compared to uniform sampling.

\paragraph{Outliers}
We performed similar experiments as above. Figure \ref{fig:sample_complexity1} (c) shows the sample complexity vs accuracy trade off curves of our algorithm and uniform sampling. Uniform sampling has a very large error compared to our algorithm, as expected. While the relative error of our algorithm decreases smoothly as the sample size grows, uniform sampling has the same high relative error. In fact, the relative error will stay high until the outlier is sampled, which  typically requires $\Omega(n)$ samples.
 
\paragraph{Large Scale}
We consider two modifications to our algorithm to optimize the performance of $\CrudeNN$ on the two challenging datasets that we are using; namely, note that both datasets are standard for benchmarking billion-scale nearest neighbor search. First, in the $\CrudeNN$ algorithm, when computing $\bD_a$ for $a\in A$, we search through the hash buckets  $h_1(a),h_2(a),\dots$ containing $a$ in increasing order of $i$ (i.e., smallest scale first), and retrieve the first $W$ (window size) distinct points in $B$ from these buckets. Then, the whole process is repeated $k$ times, with $k$ independent LSH data structures, and $\bD_a$ is set to be the distance from $a$ to the closest among all $Wk$ retrieved points.
 
Note that previously, for our smaller datasets, we set $\bD_a$ to be the distance to the first point in $B$ colliding with $a$, and repeated the LSH data structure once, corresponding to $W=k=1$. In our figures, we refer to these parameter choices as $k \times W$ and test our algorithm across several choices.

For the DEEP and Turing datasets, Figures \ref{fig:sample_complexity1} (d) and \ref{fig:sample_complexity1} (e) show the sample complexity vs relative error trade-offs for the best parameter choice (both $64 \times 10^6$) compared to uniform sampling. Qualitatively, we observe the same behavior as before: importance sampling requires fewer samples to obtain the same accuracy as uniform sampling. Regarding the other parameter choices, we see that, as expected, if we decrease $k$ (the number of LSH data structures), or if we decrease $W$ (the window size), the quality of the approximations $\{\bD_a\}_{a \in A}$ decreases and importance sampling has worse sample complexity trade-offs. Nevertheless, for all parameter choices, we see that we obtain superior sample complexity trade-offs compared to uniform sampling, as shown in Figure \ref{fig:more_largescale}. A difference between these parameter choices are the runtimes required to construct the approximations $\{\bD_a\}_{a \in A}$. For example for the DEEP dataset, the naive brute force approach (which is also optimized using parallelization) took approximately $1.3\cdot 10^4$ seconds, whereas the most expensive parameter choice of $64 \times 10^6$ took approximately half the time at $6.4 \times 10^3$ and the cheapest parameter choice of $8 \times 10^5$ took $225$ seconds, leading to a \textbf{2x-50x} factor speedup. The runtime differences between brute force and our algorithm were qualitative similar for the Turing dataset. 

Similar to the small scale dataset, our method also outperforms uniform sampling in terms of runtime if we require they both output high quality approximations. If we measure the runtime to get a $1\%$ relative error, the $16\times 2 \cdot 10^5$ version of our algorithm for the DEEP dataset requires approximately $980$ samples with total runtime approximately $1785$ seconds, whereas uniform sampling requires $> 1750$ samples and runtime $>2200$ seconds, which is $>23\%$ slower. The gap in runtime increases if we desire approximations with even smaller relative error, as the overhead of obtaining the approximations $\{\bD_a\}_{a \in A}$ becomes increasingly overwhelmed by the time needed to compute the exact answer for our samples.

\paragraph{Additional Experimental Results}
We perform additional experiments to show the utility of our approximation algorithm for the Chamfer distance for downstream tasks.
For the ShapeNet dataset, we show we can efficiently recover the true exact nearest neighbor of a fixed point cloud $A$ in Chamfer distance among a large collect of different point clouds. In other words, it is beneficial for finding the `nearest neighboring point cloud'. Recall the ShapeNet dataset,  contains approximately $5 \cdot 10^4$ different point clouds. We consider the following simple (and standard) two step pipeline: (1) use our algorithm to compute an approximation of the Chamfer distance from $A$ to every other point cloud $B$ in our dataset. More specifically, compute an approximation to $\CH(A,B) + \CH(B,A)$ for all $B$ using $50$ samples and the same parameter configurations as the small scale experiments.  Then filter the dataset of points clouds and prune down to the top $k$ closest point cloud candidates according to our approximate distances. (2) Find the closest point cloud in the top $k$ candidates via exact computation. 

We measure the accuracy of this via the standard recall $@k$ measure, which computes the fraction of times the \emph{exact} nearest neighbor $B$ of $A$, averaged over multiple $A$'s, is within the top $k$ choices.  Figure \ref{fig:more_shapenet} (a) shows that the true exact nearest neighbor of $A$, that is the point cloud $B$ which minimizes $\CH(A,B) + \CH(B,A)$ among our collection of multiple point clouds, is within the top $30$ candidates $>98\%$, time (averaged over multiple different choices of $A$). This represents a more than $\textbf{1000x}$ reduction in the number of point clouds we do exact computation over compared to the naive brute force method, demonstrating the utility of our algorithm for downstream tasks.

\begin{figure}[!ht]
    \centering
    \subfloat[\centering ShapeNet NNS pipeline experiments]{{\includegraphics[width=5cm]{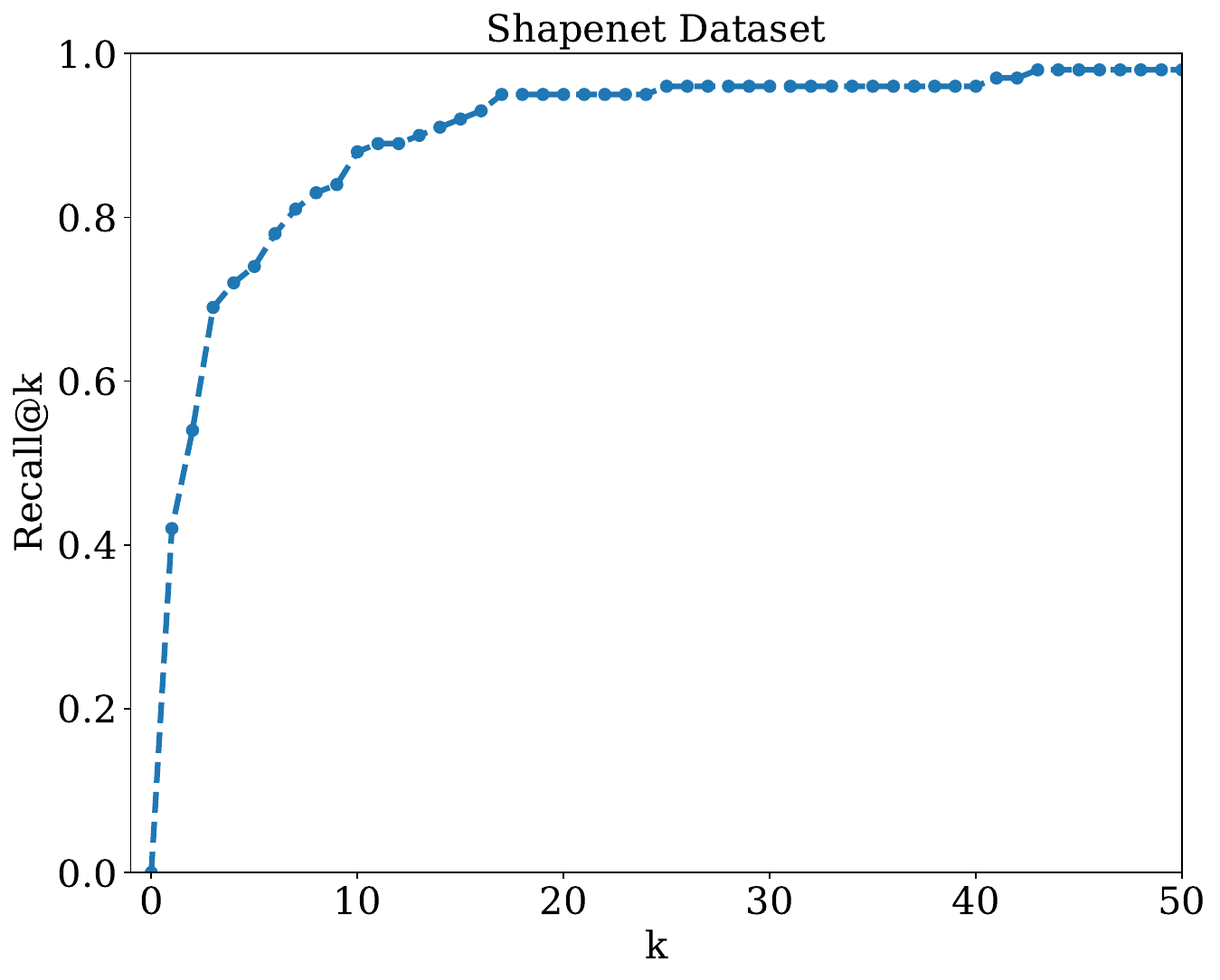} }}%
    \qquad
    \subfloat[\centering Quality of approximations $\bD_a$ vs the number of levels of LSH data structure]{{\includegraphics[width=5cm]{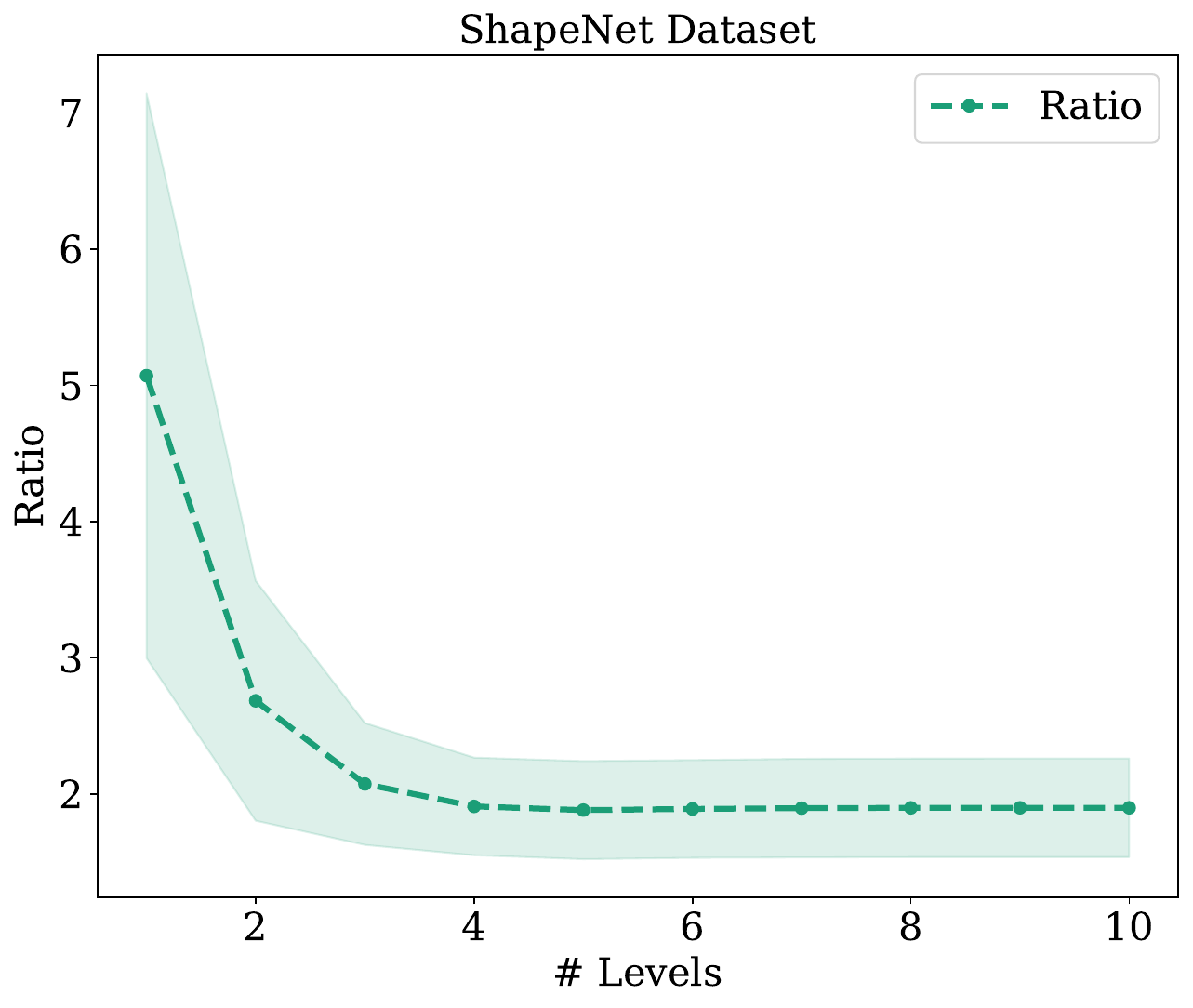} }}%
    \caption{Additional figures for the ShapeNet dataset.}%
    \label{fig:more_shapenet}%
\end{figure}

\section{Lower Bound for Reporting the Alignment}\label{sec:lowerbound}

We presented an algorithm that, in time $\bigO{nd \log(n)/\eps^2}$, produces a $(1+\eps)$-approximation to $\CH(A,B)$. It is natural to ask whether it is also possible to {\em report} a mapping $g:A \to B$ whose cost $\sum_{a \in A} \|a-g(a)\|_1$ is within a factor of $1+\eps$ from $\CH(A,B)$. (Our algorithm uses on random sampling and thusdoes not give such a mapping). This section shows that, under a popular complexity-theoretic conjecture called the {\em Hitting Set Conjecture}~\cite{williams2018some}, such an algorithm does not exists. For simplicity, we focus on the case when the underlying metric $d_X$ is induced by the Manhattan distance, i.e., $d_X(a,b)= \|a-b\|_1$. The argument is similar for the Euclidean distance, Euclidean distance squared, etc.  To state our result formally, we first define the Hitting Set (HS) problem. 

\begin{definition}[Hitting Set (HS) problem]
    The input to the problem consists of two sets of vectors $A, B \subseteq \{0,1\}^d$, and the goal is to determine whether there exists some $a \in A$ such that $a \cdot b \neq 0$ for every $b \in B$. If such an $a \in A$ exists, we say that $a$ {\em hits} $B$.

\end{definition}

It is easy to see that the Hitting Set problem can be solved in time $\bigO{n^2 d}$. The Hitting Set Conjecture~\cite{williams2018some} postulates that this running time is close to the optimal. Specifically:

\begin{conjecture}
\label{conj:hitting-set}
Suppose $d = \Theta(\log^2 n)$. Then for every constant $\delta > 0$, no randomized algorithm can solve the Hitting Set problem in $\bigO{n^{2-\delta}}$ time.
\end{conjecture}

Our result can be now phrased as follows.

\begin{theorem}[Hardness for reporting a mapping]
    Let $T(N,D,\eps)$ be the running time of an algorithm $ALG$ that, given sets of $A",B" \subset \{0,1\}^D$ of sizes at most $N$, reports a mapping $g:A" \to B"$ with cost  $(1+\eps)\textup{\CH(}A",B")$, for $D = \Theta(\log^2 N)$ and $\eps =\frac{\Theta(1)}{D}$. Assuming the Hitting Set Conjecture, we have that $T(N,D,\eps)$ is at least $\Omega(N^{2-\delta})$ for any constant $\delta>0$.
\end{theorem}

\section{Conclusion}
We present an efficient approximation algorithm for estimating the Chamfer distance up to a $1+\varepsilon$ factor in time $\bigO{nd \log(n)/\varepsilon^2}$. The result is complemented with a conditional lower bound which shows that reporting a Chamfer distance mapping of similar quality requires nearly quadratic time. Our algorithm is easy to implement in practice and compares favorably to brute force computation and uniform sampling. We envision our main tools of obtaining fast estimates of coarse nearest neighbor distances combined with importance sampling can have additional applications in the analysis of high-dimensional, large scale data.

\bibliographystyle{alpha}
\bibliography{main.bib}
\newpage
\appendix

\section{Deferred Analysis from Section~\ref{sec:algorithm}}\label{app:upper-bound}
\begin{proof}[Proof of Lemma~\ref{lem:low-variance-estimate}]
The proof follows from a standard analysis of importance sampling. The fact that our estimator $\boldeta$ is unbiased holds from the definition of $\boldeta_{\ell}$, since we are re-weighting samples according to the probability with which they are sampled in $\calD$ (in particular, the estimator is unbiased for all distributions $\calD$ where $\bD_a > 0$ for all $a \in A$). The bound on the variance is then a simple calculation:
\begin{align*}
\Varx\left[ \boldeta \right] &\leq \frac{1}{T} \cdot \left(\left[ \sum_{a \in A} \left(\frac{\bD}{\bD_a} \right) \min_{b \in B} \|a - b\|_2^2\right] - \CH(A,B)^2\right) \\
    &\leq \frac{1}{T} \cdot \left[ \sum_{a \in A} \min_{b \in B} \|a-b\|_2 \cdot \bD\right] - \frac{\CH(A,B)^2}{T}\\
	&\leq \frac{1}{T} \cdot \CH(A,B)^2 \left( \frac{\bD}{\CH(A,B)} - 1\right) .
\end{align*}
The final probability bound follows from Chebyshev's inequality.
\end{proof}

\paragraph{Locality Sensitive Hashing at every scale.}
We now discuss how to find such partitions. For the $\ell_1$ distance, each partition $i$ is formed by imposing a (randomly shifted) grid of side length $2^i$ on the dataset. Note that while the grid partitions the entire space $\R^d$ into infinitely many components, we can efficiently enumerate over the \emph{non empty} components which actually contain points in our dataset. To this end, we introduce the following definition:


\begin{definition}[Hashing at every scale]
\label{def:lsh-all-scale}
There exists a fixed constant $c_1 > 0$ and a parameterized family $\calH(r)$ of functions from $X$ to some universe $U$ such that for all $r > 0$, and for every $x, y \in X$ 
\begin{enumerate}
    \item Close points collide frequently: \begin{equation*}
        \Prx_{\bh \sim \calH(r)}\left[ \bh(x) \neq \bh(y) \right] \leq \frac{\|x -y \|_1}{r}, 
    \end{equation*}
    \item Far points collide infrequently: \begin{equation*}
\Prx_{\bh \sim \calH(r)} \left[ \bh(x) = \bh(y) \right] \leq \exp\left(- c_1 \cdot \frac{ \| x - y\|_1 }{r}\right).\end{equation*}
\end{enumerate}
\end{definition}

We are now ready to make this approach concrete via the following lemma:


\begin{lemma}[Oversampling with bounded Aspect Ratio]
\label{lem:log-n-oversampling-distribution}
Let $(X, d_X)$ be a metric space with a locality-sensitive hash family at every scale (see Definition~\ref{def:lsh-all-scale}). Consider two subsets $A, B \subset X$ of size at most $n$ and $\eps \in (0, 1)$ satisfying
\[ 1 \leq \min_{\substack{a \in A, b \in B \\ a \neq b}} d_X(a,b) \leq \max_{\substack{a \in A, b \in B}} d_X(a,b) \leq \poly(n/\eps). \]
Algorithm~\ref{fig:crude-nn}, $\emph{\CrudeNN}(A,B)$, outputs a list of (random) positive numbers $\{ \bD_a \}_{a \in A}$ which satisfy the following two guarantees:
\begin{itemize}
\item With probability $1$, every $a \in A$ satisfies $\bD_a \geq \min_{b \in B} d_X(a, b)$. 
\item For every $a \in A$,  $\Ex[\bD_a] \leq \bigO{\log n} \cdot \min_{b\in B}d_X(a,b)$.
\end{itemize}
Further, Algorithm~\ref{fig:crude-nn}, runs in time $\bigO{dn\log(n/\eps)}$ time, assuming that each function used in the algorithm can be evaluated in $\bigO{d}$ time.
\end{lemma}

Finally, we show that it always suffices to assume bounded aspect ratio:

\begin{lemma}[Reduction to bounded Aspect Ratio]
\label{lem:polynomial-aspect-ratio-suffices}
Given an instance $A, B\subset \mathbb{R}^d$ such that $|A|, |B| \le n$, and $0<\eps<1$ there exists an algorithm that runs in time $\bigO{nd\log(n)/\eps^2}$ and outputs a partition $A_1, A_2, \ldots A_T $ of $A$ and $B_1, B_2, \ldots B_T$  of $B$ such that $T= \bigO{n}$
and for each $t \in [T]$, 
\begin{equation*}
    1 \leq \min_{\substack{a \in A_t , b \in B_t \\ a \neq b}} \norm{a-b}_1 \leq \max_{\substack{a \in A_t, b \in B_t}}  \norm{a-b}_1 \leq \poly(n/\eps).
\end{equation*}
Further, 
\begin{equation*}
    \Paren{1-\eps}\textup{\CH(}A,B)   \leq \sum_{t \in [T]} \textup{\CH(}A_t,B_t) \leq \Paren{1+\eps}\textup{\CH(}A,B).
\end{equation*}
\end{lemma}

We defer the proofs of Lemma~\ref{lem:log-n-oversampling-distribution} and Lemma~\ref{lem:polynomial-aspect-ratio-suffices}
 to sub-sections~\ref{subsec:analysis-crude-ann} and~\ref{subsec:reduction-to-bounded-aspect-ratio} respectively. 
 We are now ready to complete the proof of Theorem~\ref{thm:estimating-chamfer-nearly-linear}.

\begin{proof}[Proof of Theorem~\ref{thm:estimating-chamfer-nearly-linear}]
Observe, by Lemma~\ref{lem:polynomial-aspect-ratio-suffices}, we can partition the input into pairs $(A_t, B_t)_{t\in[T]}$ such that each pair has aspect ratio at most $\poly(n/\eps)$ and the $\CH(A,B)$ is well-approximated by the direct sum of $\CH(A_t, B_t)$. Next, repeating the construction from Lemma~\ref{lem:log-n-oversampling-distribution}, and applying Markov's inequality, we have list $\{\bD_a\}_{a\in A}$ such that with probability at least $99/100$, for all $a\in A$,  $\bD_a \geq \min_{b\in B} \|a -b \|_1$ and $\bD = \sum_{a\in A} \bD_a \leq \bigO{\log(n)} \CH\Paren{A,B}$. Invoking Lemma~\ref{lem:low-variance-estimate} with the aforementioned parameters, and $T = \bigO{ \log(n)/\eps^2 }$ suffices to obtain an estimator $\eta$ which is a $(1\pm \eps)$ relative-error approximation to $\CH(A,B)$. Since we require computing the exact nearest neighbor for at most $\bigO{\log(n)/\eps^2}$ points, the running time is dominated by $\bigO{nd\log(n)/\eps^2}$, which completes the proof.
\end{proof}

\subsection{Analysis for $\CrudeNN$}
\label{subsec:analysis-crude-ann}

In this subsection, we focus analyze the $\CrudeNN$ algorithm and provide a proof for Lemma~\ref{lem:log-n-oversampling-distribution}. 
A construction of hash family satisfying Definition~\ref{def:lsh-all-scale} is given in Section~\ref{ss:lsh}. Each function from the family can be evaluated in $\bigO{d}$ time per point. We are now ready to prove Lemma~\ref{lem:log-n-oversampling-distribution}.



\begin{proof}[Proof of Lemma~\ref{lem:log-n-oversampling-distribution}] 
We note that the first item is trivially true, since $\CrudeNN(A,B)$ always sets $\bD_a$ to be some distance between $a$ and a point in $B$.  Thus, this distance can only be larger than the true minimum distance. The more challenging aspect is obtaining an upper bound on the expected value of $\bD_a$. Consider a fixed setting of $a \in A$, and the following setting of parameters:
\[ b = \argmin_{b' \in B} d_X(a, b') \qquad \gamma_a = d_X(a,b) \qquad i_0 = \lceil \log_2 \gamma_a \rceil,\]
and notice that since $\gamma_a$ is between $1$ and $\poly(n/\eps)$, we have $i_0$ is at least $0$ and at most $L = \bigO{\log (n/\eps)}$. We will upper bound the expectation of $\bD_a$ by considering a parameter $c > 1$ (which will later be set to $\bigO{\log n})$), and integrating over the probability that $\bD_a$ is at least $\gamma$, for all $\gamma \geq c \cdot \gamma_a$:
\begin{equation}
\label{eq:exp-expression}
\begin{split}
\Ex\left[ \bD_a \right] 
& \leq c \cdot \gamma_a + \int_{c\gamma_a}^{\infty} \Prx\left[ \bD_a \geq \gamma\right] d\gamma.
\end{split}
\end{equation}
We now show that for any $a \in A$, the probability that $\bD_a$ is larger than $\gamma$ can be appropriately bounded. Consider the following two bad events.
\begin{itemize}
\item $\bE_1(\gamma)$: This event occurs when there exists a point $b' \in B$ at distance at least $\gamma$ from $a$ and there exists an index $i \leq i_0$ for which $\bh_i(a) = \bh_i(b')$. 
\item $\bE_2(\gamma)$: This event occurs when there exists an index $i > i_0$ such that:
\begin{itemize}
\item For every $i' \in \{ i_0, \dots, i-1\}$, we have $\bh_{i'}(a) \neq \bh_{i'}(b)$ for all $b \in B$.
\item There exists $b' \in B$ at distance at least $\gamma$ from $a$ where $\bh_{i}(a) = \bh_{i}(b')$. 
\end{itemize}
\end{itemize}
We note that whenever $\CrudeNN(A,B)$ set $\bD_a$ larger than $\gamma$, one of the two events, $\bE_1(\gamma)$ or $\bE_2(\gamma)$, must have been triggered. 
To see why, suppose $\CrudeNN(A,B)$ set $\bD_a$ to be larger than $\gamma$ because a point $b' \in B$ with $d_X(a,b') \geq \gamma$ happened to have $\bh_i(a) = \bh_i(b')$, for an index $i \in \{0, \dots, L\}$, and that the index $i$ was the first case where it happened. If $i \leq i_0$, this is event $\bE_1(\gamma)$.
If $i > i_0$, we claim event $\bE_2(\gamma)$ occurred: in addition to $\bh_i(a) = \bh_i(b')$, it must have been the case that, for all $i' \in \{i_0, \dots, i - 1\}$, $\bh_{i'}(a) \neq \bh_{i'}(b)$ (otherwise, $i$ would not be the first index). We will upper bound the probability that either event $\bE_1(\gamma)$ or $\bE_2(\gamma)$ occurs. We make use of the tail bounds as stated in Definition \ref{def:lsh-all-scale}. The upper bound for the probability that $\bE_1(\gamma)$ is simple, since it suffices to union bound over at most $n$ points at distance larger than $\gamma$, using the fact that $r_{i_0} = 2^{i_0}$ is at most $2\cdot \gamma_a$:
\begin{align}
\Prx\left[ \bE_1(\gamma) \right] \leq n \cdot \exp\left( - c_1 \cdot \frac{\gamma}{2\gamma_a} \right). \label{eq:e-1}
\end{align}
We will upper bound the probability that event $\bE_2(\gamma)$ a bit more carefully. We will use the fact that for all $i$, the parameter $r_{i}$ is always between $2^{i - i_0} \gamma_a$ and $2^{i-i_0 + 1} \gamma_a$.
\begin{align}
\Prx\left[ \bE_2(\gamma) \right] &\leq \sum_{i > i_0} \left(\prod_{i'=i_0}^{i-1} \frac{\gamma_a}{r_{i'}} \right) \cdot \max\left\{ n \cdot \exp\left( - c_1 \cdot \frac{\gamma}{r_i}\right), 1\right\} \nonumber \\
				&\leq \sum_{i > i_0} 2^{- (0 + \dots + (i-1-i_0))}\max\left\{ \exp\left( \ln(n) -c_1 \cdot \frac{\gamma}{2^{i - i_0 + 1} \cdot \gamma_a}\right) , 1 \right\} \nonumber \\
				&\leq \sum_{k \geq 0} 2^{-\Omega(k^2)} \cdot \max\left\{ \exp\left( \ln(n) - c_1 \cdot \frac{\gamma}{2^{k+2} \cdot \gamma_a}\right), 1 \right\}. \label{eq:e-2}
\end{align}
With the above two upper bounds in place, we upper bound (\ref{eq:exp-expression}) by dividing the integral into the two contributing summands, from $\bE_1(\gamma)$ and $\bE_2(\gamma)$, and then upper bounding each individually. Namely, we have
\begin{align*}
\int_{\gamma:c \gamma_a}^{\infty} \Prx\left[ \bD_a \geq \gamma\right]d\gamma &\leq \int_{\gamma:c\gamma_a}^{\infty} \Prx\left[ \bE_1(\gamma) \right] d\gamma + \int_{\gamma:c\gamma_a}^{\infty} \Prx\left[ \bE_2(\gamma) \right]d\gamma.
\end{align*}
The first summand can be simply upper bounded by using the upper bound from (\ref{eq:e-1}), where we have
\begin{align*}
\int_{\gamma:c\gamma_a}^{\infty} \Prx\left[ \bE_1(\gamma) \right] d\gamma \leq \int_{\gamma:c\gamma_a}^{\infty} n \exp\left( - c_1 \cdot \frac{\gamma}{2 \gamma_a}\right) d\gamma \leq \frac{n \cdot 2\gamma_a}{c_1} \cdot e^{-c_1 c / 2} \leq \gamma_a
\end{align*}
for a large enough $c = \Theta(\log n)$. The second summand is upper bounded by the upper bound in (\ref{eq:e-2}), while being slightly more careful in the computation. In particular, we first commute the summation over $k$ and the integral; then, for each $k \geq 0$, we define
\[ \alpha_k := 2^{k+3} \ln(n) \cdot \gamma_a / c_1, \]
and we break up the integral into the interval $[c \cdot\gamma_a, \alpha_k]$ (if $\alpha_k < c \gamma_a$, the interval is empty), as well as $[\alpha_k, \infty)$:
\begin{align*}
\int_{\gamma:c\gamma_a}^{\infty} \Prx\left[ \bE_2(\gamma) \right] d\gamma &\leq \sum_{k\geq 0} 2^{-\Omega(k^2)} \int_{\gamma:c \gamma_a}^{\infty} \max\left\{ \exp\left( \ln(n) - c_1 \cdot \frac{\gamma}{2^{k+2} \cdot \gamma_a} \right), 1 \right\} d\gamma \\
	&\leq \sum_{k\geq 0} 2^{-\Omega(k^2)} \left( \left(\alpha_k - c \cdot \gamma_a \right)^+ + \int_{\gamma:\alpha_k}^{\infty} \exp\left( -\frac{c_1}{2} \cdot \frac{\gamma}{2^{k+2}\gamma_a}\right)d\gamma\right),
\end{align*}
where in the second inequality, we used the fact that the setting of $\alpha_k$, the additional $\ln(n)$ factor in the exponent can be removed up to a factor of two. Thus,
\begin{align*}
\int_{\gamma:c\gamma_a}^{\infty} \Prx\left[ \bE_2(\gamma) \right]d\gamma \leq \sum_{k\geq 0} 2^{-\Omega(k^2)} \left( \gamma_a \cdot \bigO{2^k\log n} + \gamma_a \cdot \bigO{2^{k}}\right) = \bigO{\log n} \cdot \gamma_a.
\end{align*}

Finally, the running time is dominated by the cost of evaluating $\bigO{\log (n/\eps)}$ functions on $n$ points in dimension $d$. Since each evaluation takes $\bigO{d}$ time, the bound follows.
\end{proof}

\subsection{Locality-Sensitive Hashing at Every Scale}
\label{ss:lsh}

\begin{lemma}[Constructing a LSH at every scale]\label{lem:ell-1-hash}
For any $r \geq 0$ and any $d \in \N$, there exists a hash family $\calH(r)$ such that for any two points $x, y \in \R^d$, 
\begin{align*}
\Prx_{\bh \sim \calH(r)}\left[ \bh(x) \neq \bh(y)\right] &\leq \frac{\|x-y\|_1}{r} \\
\Prx_{\bh \sim \calH(r)}\left[ \bh(x) = \bh(y) \right] &\leq \exp\left( - \frac{\|x- y\|_1}{r}\right).
\end{align*}
In addition, for any $\bh \sim \calH(r)$, $\bh(x)$ may be computed in $\bigO{d}$ time.
\end{lemma}

\begin{proof}
The construction proceeds in the following way: in order to generate a function $\bh \colon \R^d \to \Z^d$ sampled from $\calH(r)$,
\begin{itemize}
\item We sample a random vector $\bz \sim [0, r]^d$.
\item We let
\[ \bh(x) = \left( \left\lceil \frac{x_1 + \bz_1}{r} \right\rceil, \left\lceil \frac{x_2 + \bz_2}{r} \right\rceil, \dots , \left\lceil \frac{x_d + \bz_d}{r} \right\rceil \right). \]
\end{itemize}
Fix $x, y \in \R^d$. If $\bh(x) \neq \bh(y)$, there exists some coordinate $k \in [d]$ on which $\bh(x)_k \neq \bh(y)_k$. This occurs whenever (i) $|x_k - y_k| > r$, or (ii) $|x_k - y_k|\leq r$, but $\bz_k$ happens to fall within an interval of length $|x_k - y_k|$, thereby separating $x$ from $y$. By a union bound,
\begin{align*}
\Prx_{\bh \sim \calH(r)}\left[ \bh(x) \neq \bh(y) \right] &\leq \sum_{k=1}^d \frac{|x_k - y_k|}{r} = \frac{\|x-y\|_1}{r}.  
\end{align*}
On the other hand, in order for $\bh(x) = \bh(y)$, it must be the case that every $|x_k - y_k| \leq r$, and in addition, the threshold $\bz_k$ always avoids an interval of length $|x_k - y_k|$. The probability that this occurs is
\begin{align*}
\Prx_{\bh \sim \calH(r)}\left[ \bh(x) = \bh(y)\right] &= \prod_{k=1}^d \max\left\{ 0, 1 - \frac{|x_k - y_k|}{r} \right\} \leq \exp\left( - \sum_{k=1}^d \frac{|x_k-y_k|}{r}\right) \\
	&\leq \exp\left(-\frac{\|x-y\|_1}{r}\right).
\end{align*}
\end{proof}

Extending the above construction to $\ell_2$ follows from embedding the points $A \cup B$ into $\ell_1$ via a standard construction. 

\begin{theorem}[\cite{matousek2013lectures}]\label{thm:l2embed}
Let $\eps \in (0,1)$ and define $\bT: \R^d \rightarrow \R^k$ by 
\[\bT(x)_i = \frac{1}{\beta k} \sum_{j=1}^d Z_{ij}x_j, \quad i= 1, \ldots, k \]
where $\beta = \sqrt{2/\pi}$. Then for every vector $x \in \R^d$, we have
\[\Pr[(1-\eps)\|x\|_2 \le \|\bT(x)\|_1 \le (1+\eps) \|x\|_2] \ge 1-e^{c \eps^2 k}, \]
where $c > 0$ is a constant.
\end{theorem}

The map $\bT \colon \R^d \to \R^k$ with $k = \bigO{\log n/\eps^2}$ gives an embedding of $A \cup B$ into $\ell_1^k$ of distortion $(1\pm \eps)$ with high probability. Formally, with probability at least $1 - 1/n$ over the draw of $\bT$ with $t = \bigO{\log n/\eps^2}$, every $a \in A$ and $b \in B$ satisfies
\begin{align*}
(1-\eps) \|a - b\|_2 \leq \| \bT(a) - \bT(b) \|_1 \leq (1+\eps)\| a - b\|_2. 
\end{align*}
This embedding has the effect of reducing $\ell_2$ to $\ell_1$ without affecting the Chamfer distance of the mapped points by more than a $(1\pm \eps)$-factor. In addition, the embedding incurs an extra additive factor of $\bigO{nd \log n/\eps^2}$ to the running time in order to perform the embedding for all points.

\subsection{Reduction to $\poly(n/\eps)$ Aspect Ratio for $\ell_p$, $p \in [1,2]$}
\label{subsec:reduction-to-bounded-aspect-ratio}

In this section, we discuss how to reduce to the case of a $\poly(n/\eps)$ aspect ratio. The reduction proceeds by first obtaining a very crude estimate of $\CH(A, B)$ (which will be a $\poly(n)$-approximation), applying a locality-sensitive hash function in order to partition points of $A$ and $B$ which are significantly farther than $\poly(n) \cdot \CH(A, B)$. Finally, we add $\bigO{\log n}$ coordinates and add random vector of length $\poly(\eps/n) \cdot \CH(A, B)$ in order to guarantee that the minimum distance is at least $\poly(\eps /n) \cdot \CH(A, B)$ without changing $\CH(A,B)$ significantly.

\begin{proof}[Proof of Lemma~\ref{lem:polynomial-aspect-ratio-suffices}]

\textbf{Partitioning Given Very Crude Estimates.} In particular, suppose that with an $\bigO{nd} + \bigO{n\log n}$ time computation, we can achieve a value of $\boldeta \in \R_{\geq 0}$ which satisfies 
\[ \CH(A, B) \leq \boldeta \leq c\cdot \CH(A, B), \]
with high probability (which we will show how to do briefly with $c = \poly(n)$). Then, consider sampling $\bh \sim \calH(c n \cdot \boldeta)$ and partitioning $A$ and $B$ into equivalence classes according to where they hash to under $\bh$. The probability that two points at distance farther than $\bigO{\log n} \cdot cn\cdot \boldeta$ collide under $\bh$ is small enough to union bound over at most $n^2$ many possible pairs of vectors. In addition, the probability that there exists $a \in A$ for which $b \in B$ minimizing $\|a - b\|_p$ satisfies $\bh(a) \neq \bh(b)$ is at most $\CH(A, B) / (cn \cdot \boldeta) \leq 1/n$. This latter inequality implies that computing the Chamfer distance of the corresponding parts in the partition and summing them is equivalent to computing $\CH(A, B)$. 

\textbf{Getting Very Crude Estimates.} We now show how to obtain a $\poly(n)$-approximation to $\CH(A, B)$ in time $\bigO{nd} + \bigO{n \log n}$ for points in $\R^d$ with $\ell_p$ distance. This is done via the $p$-stable sketch of Indyk~\cite{indyk2006stable}. In particular, we sample a vector $\bg \in \R^d$ by independent $p$-stable random variables (for instance, $\bg$ is a standard Gaussian vector for $p = 2$ and a vector of independent Cauchy random variables for $p = 1$). We may then compute the scalar random variables $\{ \langle a, \bg \rangle \}_{a \in A}$ and $\{ \langle b, \bg \rangle \}_{b \in B}$, which give a projection onto a one-dimensional space. By $p$-stability, for any $a \in A$ and $b \in B$, the distribution of $\langle a, \bg \rangle - \langle b, \bg \rangle$ is exactly as $\| a - b\|_p \cdot \bg'$, where $\bg'$ is an independent $p$-stable random variable. Hence, we will have that for every $a \in A$ and $b \in B$, 
\[ \frac{\| a - b\|_p}{\poly(n)} \leq \left| \langle a , \bg \rangle - \langle b, \bg \rangle \right| \leq \|a - b\|_p \cdot \poly(n), \]
with probability $1-1/\poly(n)$ and hence $\CH(\{ \langle a, \bg \rangle \}_{a \in A}, \{ \langle b , \bg \rangle \}_{b \in B})$, which is computable by 1-dimensional nearest neighbor search (i.e., repeatedly querying a binary search tree), gives a $\poly(n)$-approximation to $\CH(A,B)$.

\textbf{Adding Distance} Finally, we now note that $\boldeta / c$ gives us a lower bound on $\CH(A, B)$. Suppose we append $\bigO{\log n}$ coordinates to each point and in those coordinates, we add a random vector of norm $\eps \cdot \boldeta / (cn)$. With high probability, every pair of points is now at distance at least $\eps \cdot \boldeta / (cn)$. In addition, the Chamfer distance between the new set of points increases by at most an additive $\bigO{\eps \boldeta / c}$, which is at most $\bigO{\eps} \cdot \CH(A, B)$, proving Lemma \ref{lem:polynomial-aspect-ratio-suffices}.


\end{proof}

\section{Deferred Analysis from Section~\ref{sec:lowerbound}}\label{app:lower-bound}

\begin{proof}

To set the notation, we let $n_A=|A|$, $n_B=|B|$. 

The proof mimics the argument from~\cite{rohatgi2019conditional}, which proved a similar hardness result for the problem of computing the Earth-Mover Distance. In particular, Lemma 4.3 from that paper shows the following claim.

\begin{claim} 
\label{c:map}
For any two sets $A, B \subseteq \{0,1\}^d$, there is a mapping $f:\{0,1\}^d \to \{0,1\}^{d"}$ and a vector  $v \in \{0,1\}^{d"}$, such that $d"=\bigO{d}$ and for any $a \in A$, $b \in B$:
\begin{itemize}
    \item If $a \cdot b=0$ then $\|f(a)-f(b)\|_1 = 4d+2$,
    \item If $a \cdot b>0$ then $\|f(a)-f(b)\|_1 \ge 4d+4$,
    \item $\|f(a)-v\|_1 = 4d+4$.
\end{itemize}
Furthermore, each evaluation $f(a)$ can be performed in $\bigO{d}$ time.
\end{claim}

We will be running $ALG$ on sets $A"=\{f(a): a \in A\}$ and $B"=\{f(b): b \in B\} \cup \{v\}$. It can be seen that, given a reported mapping $g$, we can assume that for all $a" \in A"$ we have $\|a"-g(a")\|_1 \le 4d+4$, as otherwise $g$ can map $a''$ to $v$. 
If for all $a \in A$ there exists $b \in B$ such that $a \cdot b=0$, i.e., $A$ does not contain a hitting vector, then the optimal mapping cost is $n_A(4d+2)$. More generally, let $H$ be the set of  vectors $a \in A$ hitting $B$, and let $h=|H|$. It can be seen that 
\[\CH(A",B") = h (4d+4) + (n_A -h)(4d+2) = n_A(4d+2) + 2h.\]  

Thus, if we could compute $\CH(A",B")$ exactly, we would determine if $h=0$ and solve HS. In what follows we show that even an approximate solution can be used to accomplish this task as long as $\eps$ is small enough.

Let $t=c \log (n)/\eps$ for some large enough constant $c>1$. Consider the  algorithm $HittingSet(A, B)$ that solves HS by invoking the algorithm $ALG$.

\begin{figure}[h!]
	\begin{framed}
		
		\begin{flushleft}
            \noindent Subroutine $HittingSet(A, B)$
            
			\noindent {\bf Input:} Two sets $A, B \subset \{0,1\}^d$ of size at most $n$, and an oracle access to $ALG$ that computes $(1+\eps)$-approximate $\CH$.
			
			\noindent {\bf Output:} Determines whether there exists $a \in A$ such that $a \cdot b>0$ for all $b \in B$
			
			\begin{enumerate}
                \item Sample (uniformly, without replacement) $\min(t,|A|)$ distinct vectors $a \in A$, and for each of them check if $a \cdot b>0$ for all $B$. If such an $a$ is found, return YES. 
                \item Construct $A", B"$ as in Claim~\ref{c:map}, and invoke $ALG$. Let $g: A" \to B"$ be the returned map. 
                \item Identify the set $M$ containing all $a \in A$ such that $\|g(f(a)) - f(a)\|_1=4d+2$.  Note that $a \cdot b=0$ for $b \in B$ such that $f(b)=g(f(a))$.
                \item Recursively execute $HittingSet(A-M, B)$
			\end{enumerate}
		\end{flushleft}
	\end{framed}
	\caption{Reduction from Hitting Set to $(1+\eps)$-approximate $\CH$, implemented using algorithm $ALG$ .}\label{fig:red}
\end{figure}

It can be seen that the first three steps of the algorithm take at most $\bigO{ntd}$ time. Furthermore, if the algorithm terminates, it reports the correct answer, as only vectors $a$ that are guaranteed not to be hitting are removed in the recursion. It remains to bound the total number and cost of the recursive steps. To this end, we will show that, with high probability, in each recursive call we have $|A-M| \le |A|/2$. This will yield a total time of 
$\log n [(n t d) + T(n+1, \bigO{d}, \eps)]$. 
Since $t=c \log (n)/\eps$, $d = \log^2 n$ and $\eps = \frac{\Theta(1)}{d}$, it follows that the time is at most
$n  \log^5 (n) + \log (n)  T(n+1, \bigO{d}, \eps)$, and the theorem follows.

To show that $|A-M| \le |A|/2$, first observe that if the algorithm reaches step (2), then for a large enough constant $c>1$ it holds,  with high probability, that the set $H$ of hitting vectors $a$ has cardinality  at most $\eps \cdot n_A$, as otherwise one such vector would have been sampled. Thus, the subroutine $ALG$ returns a map where the vast majority of the points $f(a)$ have been matched to a point $f(b)$ such that $\|f(a)-f(b)\|_1=4d+2$. More formally, the cost of the mapping $g$ is
\begin{eqnarray*}
    C& = &  \sum_{a" \in A"} \|a" -g(a")\|_1\\
    & \le & (1+\eps) [ n_A(4d+2) + 2 |H|] \\
    & \le &  (1+\eps) [ n_A(4d+2) + 2 \eps n_A]\\
    &\le &  n_A(4d+2) + 4 \eps n_A (d+2)\\
    &\le & n_A(4d+2) + n_A 
\end{eqnarray*}  
where in the last step we used the assumption about $\eps$.

Denote $m=|M|$. Observe that the cost $C$ of the mapping $g$ can be alternatively written as:
\[  C = m(4d+2) + (n_A-m)(4d+4) = n_A (4d+4)-2m \]
This implies $m=(n_A (4d+4) - C)/2$. Since we showed earlier that $C \le  n_A (4d+2)+n_A$, we conclude that
\[ m=(n_A (4d+4) - C)/2 \ge (2 n_A - n_A)/2 = n_A/2 . \]

Thus, $|A-M|=n_A-m \le n_A/2$, completing the proof.

\end{proof}

\end{document}